\documentclass[twocolumn]{IEEEtran}
\pdfoutput=1
\usepackage{amsthm,amssymb,graphicx,multirow,amsmath,color,amsfonts}
\usepackage[update,prepend]{epstopdf}
\usepackage[noadjust]{cite}
\usepackage[latin1]{inputenc}
\usepackage{tikz}
\usetikzlibrary{arrows,calc}		
\usepackage{bbm} 
\usepackage{pdfpages}
\usepackage{flushend}
\usepackage{tabulary}
\usepackage{multirow}
\usepackage{comment}

\catcode`,\active

\catcode`\,12

\def\nbx{{\mathbf{x}}}

\def\nb0{{\mathbf{0}}}
\def\nb1{{\mathbf{1}}}

\def\ncalA{{\mathcal{A}}}

\def\ncalE{{\mathcal{E}}}
\def\ncalF{{\mathcal{F}}}

\def\ncalK{{\mathcal{K}}}

\def\ncalN{{\mathcal{N}}}
\def\ncalO{{\mathcal{O}}}
\def\ncalP{{\mathcal{P}}}

\def\ncalR{{\mathcal{R}}}
\def\ncalS{{\mathcal{S}}}
\def\ncalT{{\mathcal{T}}}

\def\ncalX{{\mathcal{X}}}

\def\nbbE{{\mathbb{E}}}

\def\nbbP{{\mathbb{P}}}

\def\nbbR{{\mathbb{R}}}

\def\nfrakR{{\mathfrak{R}}}

\newtheorem{lemma}{Lemma}

\newtheorem{definition}{Definition}

\newtheorem{theorem}{Theorem}

\newtheorem{cor}{Corollary}

\newtheorem{remark}{Remark}

\def\PPPok{\Phi_k}							
\def\PPPo{\Phi}
\def\PPPak{\Phi_k^{\rm (a)}}				
\def\PPPa{\Phi^{\rm (a)}} 

\def\PPPaj{\Phi_j^{\rm (a)}}				

\def\densityok{\lambda_k}
\def\densityak{\lambda_k^{\rm (a)}}

\def\E{\mathbb{E}}

\def\pc{\mathtt{P_c}}
\def\rc{\mathtt{R_c}}   

\def\R{\mathbb{R}}

\def\T{\beta}							

\def\sir{\mathtt{SIR}}

\def\calA{\mathcal{A}}

\allowdisplaybreaks 
\begin{document}

\graphicspath{{./Figures/}}
\title{Fundamentals of Heterogeneous Cellular Networks with Energy Harvesting}
\author{Harpreet S. Dhillon, Ying Li, Pavan Nuggehalli, Zhouyue Pi, and Jeffrey G. Andrews\thanks{
H. S. Dhillon and J. G. Andrews are with the Wireless Networking and Communications Group (WNCG), The University of Texas at Austin, TX (email: dhillon@utexas.edu and jandrews@ece.utexas.edu). Y. Li and Z. Pi are with Samsung Research America, Richardson, TX (email: \{yli2, zpi\}@sta.samsung.com). P.~Nuggehalli is with the Mobile and Wireless Group of Broadcom Corporation, Sunnyvale, CA (email: pavann@broadcom.com). 

This work was done while the first author was with Samsung Research America, Richardson, TX. A part of this paper is accepted for presentation at IEEE Globecom 2013 in Atlanta, GA~\cite{DhiLiC2013}. \hfill Manuscript updated: \today.
} }

\maketitle

\begin{abstract}
We develop a new tractable model for $K$-tier heterogeneous cellular networks (HetNets), where each base station (BS) is powered solely by a self-contained energy harvesting module.  The BSs across tiers differ in terms of the energy harvesting rate, energy storage capacity, transmit power and deployment density. Since a BS may not always have enough energy, it may need to be kept OFF and allowed to recharge while nearby users are served by neighboring BSs that are ON.  We show that the fraction of time a $k^{th}$ tier BS can be kept ON, termed {\em availability} $\rho_k$, is a fundamental metric of interest. Using tools from random walk theory, fixed point analysis and stochastic geometry, we characterize the set of $K$-tuples $(\rho_1,\rho_2,\ldots\rho_K)$, termed the {\em availability region}, that is achievable by general uncoordinated operational strategies, where the decision to toggle the current ON/OFF state of a BS is taken independently of the other BSs. If the availability vector corresponding to the optimal system performance, e.g., in terms of rate, lies in this availability region, there is no performance loss due to the presence of unreliable energy sources. As a part of our analysis, we model the temporal dynamics of the energy level at each BS as a birth-death process, derive the energy utilization rate, and use hitting/stopping time analysis to prove that there exists a fundamental limit on $\rho_k$ that cannot be surpassed by any uncoordinated strategy. 
\end{abstract}

\section{Introduction}
The possibility of having a self-powered BS is becoming realistic due to several parallel trends. First, BSs are being deployed ever-more densely and opportunistically to meet the increasing capacity demand~\cite{AndJ2013}. The new types of BSs, collectively called ``small cells'', cover much smaller areas and hence require significantly smaller transmit powers compared to the conventional macrocells. Second, due to the increasingly bursty nature of traffic, the loads on the BSs will experience massive variation in space and time~\cite{DhiGanJ2013}. In dense deployments, this means that many BSs can, in principle, be turned OFF most of the time and only be requested to wake up intermittently based on the traffic demand. Third, energy harvesting techniques, such as solar power, are becoming cost-effective compared to the conventional sources~\cite{WorCenM2006}. This is partly due to the technological improvements and partly due to the market forces, such as increasing taxes on conventional power sources, and subsidies and regulatory pressure for greener techniques. Fourth, high-speed wireless backhaul is rapidly becoming a reality for small cells, which eliminates the need for other wired connections~\cite{HurKimC2011}. Therefore, being able to avoid the constraint of requiring a wired power connection is even more attractive, since it would open up entire new categories of low-cost ``drop and play'' deployments, especially of small cells.

\subsection{Related Work}
The randomness in the energy availability at a transmitter demands significant rethinking of conventional wireless communication systems. There are three main directions taken in the literature to address this challenge, which we order below in terms of complexity and realism. The first considers a relatively simple setup consisting of single full-buffer isolated link, and study optimal transmission strategies under a given energy arrival process~\cite{HoZhaJ2012,OzeTutJ2011,TutYenJ2012}. The effect of data arrivals can be additionally incorporated by considering two consecutive queues at the transmitter, one for the data and the other for the energy arrivals~\cite{ShaMukJ2010,YanUluJ2012}. 

Second, a natural extension of an isolated link, considers a broadcast channel, where a single isolated transmitter serves multiple users. Again one can assume full-buffer at the transmitter so that the transmission strategies need to be adapted only to the energy arrival process, e.g., in~\cite{YanOzeJ2012}. More realistically, one can relax the full-buffer assumption to explicitly consider data arrivals as discussed above for the isolated link, and optimize various metrics, e.g., minimize packet transmission delay~\cite{AntUysJ2011}, or maximize system throughput~\cite{GatGeoJ2010}.

The third and least investigated direction is to consider multiple self-powered transmitters, which significantly generalizes the above two directions. Generally speaking, the main goal is to adapt transmission schemes based on the energy and load variations across both time and space. While some progress has been recently made in advancing the understanding of mobile ad-hoc networks (MANETs) with self-powered nodes, see~\cite{HuaJ2013,HuaC2011} and references therein, our understanding of cellular networks in a similar setting is severely limited. This is partly due to the fact that conventional cellular networks consisted of big macro BSs that required fairly high power, and it made little sense to study them in the context of energy harvesting. As discussed earlier, this is not the case with a HetNet, which may support ``drop and play'' deployments, especially of small cells, in the future. Comprehensive modeling and analysis of this setup is the main focus of this paper.  

To capture key characteristics of HetNets, such as heterogeneity in infrastructure, and increasing uncertainty in BS locations, we consider a general $K$-tier cellular network with $K$ different classes of BSs, where the BS locations of each tier are sampled from an independent Poisson Point Process (PPP). This model was proposed for HetNets in~\cite{DhiGanC2011,DhiGanJ2012}, with various extensions and generalizations in~\cite{JoSanJ2012,MukJ2012,MadResC2011}. The model, although simple, has been validated as reasonable since then both by empirical evidence~\cite{TayDhiC2012} and theoretical arguments~\cite{BlaKarJ2012}. Due to its realism and tractability, it has become an accepted model for HetNets, see~\cite{ElSHosJ2013} for a detailed survey. 

\subsection{Contributions}

{\em Tractable and general system model.} We propose a general system model consisting of $K$ classes of self-powered BSs, which may differ in terms of the transit power, deployment density, energy harvesting rate and energy storage capacity. Due to the uncertainty in the energy availability, a BS may need to be kept OFF and allowed to accumulate enough energy before it starts serving its users again. In the meanwhile, its load is transferred to the neighboring BSs that are ON. Thus, at any given time a BS can be in either of the two {\em operational states}: ON or OFF. In this paper, we focus on {\em uncoordinated operational strategies}, where the operational state of each BS is toggled independently of the other BSs. For tractability, we assume that the network operates on two time scales: i) {\em long time scale}, over which the decision to turn a BS ON or OFF is taken, and ii) {\em short time scale}, over which the scheduling and cell selection decisions are taken. As discussed in Section~\ref{sec:availability}, this distinction facilitates analysis in two ways: a) it allows us to assume that the operational states of the BSs are static over short time scale, and b) it allows us to consider the average effects of cell selection over long time scale.

{\em Availability region.} We show that the fraction of time a $k^{th}$ tier BS can be kept in the ON state, termed the {\em availability} $\rho_k$, is a key metric for self-powered cellular networks. Using tools from random walk theory, fixed point analysis, and stochastic geometry, we characterize the set of $K$-tuples $(\rho_1, \rho_2, \ldots \rho_K)$, termed the {\em availability region}, that are achievable with a set of general uncoordinated strategies. Our analysis involves modeling the temporal dynamics of the energy level at each BS as a birth-death process, deriving energy utilization rate for each class of BSs using stochastic geometry, and using hitting/stopping time analysis for a Markov process to prove that there exists a fundamental limit on the availabilities $\{\rho_k\}$, which cannot be surpassed by any uncoordinated strategy. We also construct an achievable scheme that achieves this upper limit on availability for each class of BSs.

{\em Notion of optimality for self-powered HetNets.} The characterization of exact availability region lends a natural notion of optimality to self-powered HetNets. Our analysis concretely demonstrates that if the $K$-tuple $(\hat{\rho}_1, \hat{\rho}_2, \ldots, \hat{\rho}_K)$ corresponding to the optimal performance of the network, e.g., in terms of downlink rate, lies in the availability region, the performance of the HetNet with energy harvesting is fundamentally the same as the one with reliable energy sources. Using recent results for downlink rate distribution in HetNets~\cite{DhiAndJ2013,SinDhiJ2013}, we also show that it is not always optimal from downlink data rate perspective to operate the network in the regime corresponding to the maximum availabilities, i.e., it may be preferable to keep a certain fraction of BSs OFF despite having enough energy.

\section{System Model} \label{sec:sysmod}

\begin{table}
\centering
\caption{Notation Summary}
\label{table:notationtable}
\begin{tabulary}{\columnwidth}{ |c | C | }
\hline
    \textbf{Notation} & \textbf{Description} \\ \hline
    $\ncalK$		&	Set of indices for BS tiers, i.e., $\ncalK = \{1, 2, \ldots, K\}$ \\ \hline
    $\PPPok$, $\densityok$; $\PPPo$	& 	Independent PPP modeling locations of $k^{th}$ tier BSs, its density; set of all BSs, i.e., $\PPPo = \cup_{k \in \ncalK} \PPPok$ \\ \hline
    $\Phi_u; \lambda_u$					& An independent PPP modeling user locations, density of users \\ \hline
    $\mu_k$; $\nu_k$; $N_k$ &		Energy harvesting rate, utilization rate, and energy storage capacity of a $k^{th}$ tier BS \\ \hline
    $\rho_k; \nfrakR$ 	& 	Availability of $k^{th}$ tier BSs; availability region \\ \hline
    $\PPPak, \densityak; \PPPa$ 	&	Independent PPP modeling the $k^{th}$ tier BSs that are available, its density $\densityak = \rho_k \densityok$; all available BSs, i.e., $\PPPa = \cup_{k \in \ncalK} \PPPak$ \\ \hline
    $P_k$			& 	Downlink transmit power of a $k^{th}$ tier BS to each user in each resource block \\ \hline
    $h_{k}; \ncalX_{k}; \alpha$			&  Small scale fading gain $h_{k}\sim \exp(1)$; large scale shadowing gain (general distribution) from a $k^{th}$ tier BS; path loss exponent \\ \hline  
    $x_k^{*(z)}, x^{*(z)}$	&	Candidate serving BS in $\PPPak$ for user at $z \in \Phi_u$, serving BS for $z \in \Phi_u$ \\ \hline
    $\pc; \beta$			& Coverage probability; target $\sir$ \\ \hline
    $\rc; \ncalT$	& Rate coverage; target rate \\ \hline
\end{tabulary}

\end{table}

\subsection{System Setup and Key Assumptions}
We consider a $K$-tier cellular network consisting of $K$ different classes of BSs. For notational simplicity, define $\ncalK = \{1,2,\ldots,K\}$. The locations of the BSs of the $k^{th}$ tier are modeled by an independent PPP $\PPPok$ of density $\densityok$. Each BS has an energy harvesting module and an energy storage module, which is its sole source of energy. The BSs across tiers may differ both in terms of how fast they harvest energy, i.e., the {\em energy harvesting rate} $\mu_k$ joules/sec, and how much energy they can store, i.e., the energy storage capacity (or battery capacity) $N_k$ joules. We assume that the normalization of $\mu_k$ and $N_k$ is such that each user requires 1 joule of energy per sec. This assumption can be easily relaxed to incorporate users requiring more energy under sufficient randomization, but this case is not in the scope of this paper. For resource allocation, we assume an orthogonal partitioning of resources, e.g., time-frequency resource blocks in orthogonal frequency division multiple access (OFDMA), where each resource block is allocated to a single user. Due to orthogonal resource allocation, there is no intra-cell interference. Note that a user can be allocated multiple resource blocks as discussed in detail in the sequel. We further assume that a $k^{th}$ tier BS transmits to each user with a fixed power $P_k$ in each resource block, which may depend upon the energy harvesting parameters, although we do not study this dependence in this paper. The target $\sir$ $\T$ is the same for all the tiers.

The energy arrival process at a $k^{th}$ tier BS is modeled as a Poisson process with mean $\mu_k$. Since most energy harvesting modules contain smaller sub-modules, each harvesting energy independently, e.g., small solar cells in a solar panel, the net energy harvested can be modeled as a Binomial process, which approaches the Poisson process when the number of sub-modules grow large. Interestingly, this model has been validated using empirical measurements for a variety of energy harvesting modules~\cite{RouWriB2004}.  
Since the energy arrivals are random and the energy storage capacities are finite, there is some uncertainty associated with whether the BS has enough energy to serve users at a particular time. Under such a constraint, it is required that some of the BSs be kept OFF and allowed to recharge while their load is handled by the neighboring BSs that are ON. Besides, as discussed in the sequel, it may also be preferable to keep a BS OFF despite having enough energy. Therefore, a BS can be in either of the two {\em operational states}: ON or OFF. The decision to toggle the operational state from one to another is taken by the operational strategies that can be broadly categorized into two classes.

{\em Uncoordinated:} In this class of strategies, the decision to toggle the operational state, i.e., turn a BS ON or OFF, is taken by the BS independently of the operational states of the other BSs. For example, a BS may decide to turn OFF if its current energy level reaches below a certain predefined level and turn back ON after harvesting enough energy. The BS may additionally consider the time for which it is in the current state while making the decision. For instance, a BS may start a timer whenever the state is toggled and may decide to toggle it back when the timer expires or the energy level reaches a certain minimum value, whichever occurs first. This class will be the main focus of this paper.

{\em Coordinated:} In this class of strategies, the decision to toggle the state of a particular BS is dependent upon the states of the other BSs. For example, the BSs may be partitioned into small clusters where only a few BSs in each cluster are turned ON. The decision may be taken by some central entity based on the current load offered to the network. This is useful in the cases where the load varies by orders of magnitude across time, e.g., due to diurnal variation. A small fraction of BSs is enough to handle smaller load, with the provision of turning more ON as the load gradually increases. In addition to the load, other factors such as network topology and interference among BSs may also affect the decision.

For tractability, we define the following two time scales over which the network is assumed to operate.
\begin{definition}[Time scales]
The scheduling and cell association decisions are assumed to be taken over a time scale that is of the order of the scheduling block duration. We term this time scale as a short time scale. On the other hand, the operational policies that toggle the operational state of a BS are assumed to be defined on a much longer time scale. We will henceforth term this time scale as a long time scale. 
\end{definition}
As discussed in the sequel, this distinction is the key to tractability because of two reasons: i) it allows us to assume the energy states of the BSs to be static over short time scales, and ii) it allows us to consider the average effects of cell selection while determining the energy utilization rates over long time scales. Due to uncertainty in the energy availability or due to the optimality of a given performance metric, e.g., downlink rate, all the BSs in the network may not always be available to serve users. This is made precise by defining availability of a BS as follows.

\begin{definition}[Availability]
A BS is said to be available if it is in the ON state as a part of the operational policy and has enough energy to serve at least one user, i.e., has at least one unit of energy. The probability that a BS of tier $k$ is available is denoted by $\rho_k$, which may be different for each tier of BSs due to the differences in the capabilities of the energy harvesting modules and the load served. For notational simplicity, we denote the set of availabilities for the $K$ tiers by $\{\rho_k\}$.
\end{definition}

For uncoordinated strategies, it is reasonable to assume that the current operational state (ON or OFF) of a BS is independent of the other BSs, especially since the energy harvesting processes are assumed to be independent across the BSs. The coupling in the transmission of various BSs that arises due to interference and mobility is ignored. Under this independence assumption, the set of ON BSs of the $k^{th}$ tier form a PPP $\PPPak$ with density $\densityak = \lambda_k \rho_k$. This results from the fact that the independent thinning of a PPP leads to a PPP with appropriately scaled density~\cite{KinB1993}. As will be evident from the availability analysis in the next section, this abstraction is the key that makes this model tractable and leads to meaningful insights. 

\subsection{Propagation and Cell Selection Models}
For this discussion it is sufficient to consider only the BSs that are available, i.e., the ones that are in the ON state.  For notational ease, define $\PPPa = \cup_{k \in \ncalK} \PPPak$. The user locations are assumed to be drawn from an independent PPP $\Phi_u$ of density $\lambda_u$. More sophisticated non-uniform user distribution models~\cite{DhiGanJ2013b} can also be considered but are not in the scope of this paper. The received power at a user located at $z \in \Phi_u$ from a $k^{th}$ tier BS placed at $x_k \in \PPPak$ in a given resource block is
\begin{align}
P(z, x_k) = P_k h_{kx_k}^{(z)} \ncalX_{kx_k}^{(z)} \|x_k - z\|^{-\alpha},
\label{eq:Pow_received}
\end{align}
where $h_{kx_k}^{(z)} \sim \exp(1)$ models Rayleigh fading, $\ncalX_{kx_k}^{(z)}$ models large scale shadowing, and $\|x_k - z\|^{-\alpha}$ represents standard power-law path loss with exponent $\alpha$, for the wireless channel from $x_k \in \PPPak$ to $z \in \Phi_u$. Note that since $h_{kx_k}^{(z)}$ and $\ncalX_{kx_k}^{(z)}$ are both independent of the locations $x_k$ and $z$, we will drop $x_k$ and $z$ from the subscript and superscript, respectively, and denote the two random variables by $h_k$ and $\ncalX_k$, whenever the locations are clear from the context.

For cell selection, we assume that each user connects to the BS that provides the highest long term received power, i.e., small scale fading gain $h_{kx}^{(z)}$ does not affect cell selection. For a cleaner exposition, we denote the location of the candidate $k^{th}$ tier serving BS for $z \in \Phi_u$ by $x_k^{*(z)} \in \PPPak$, which is
\begin{align}
x_k^{*(z)} = \arg \max_{x \in \PPPak} P_k \ncalX_{xk}^{(z)} \|x-z\|^{-\alpha}.
\end{align}
A user $z \in \Phi_u$ now selects one of these $K$ candidate serving BSs based on the average received signal power, i.e., the location of the serving BS $x^{*(z)} \in \{x_k^{*(z)}\}$ is
\begin{align}
x^{*(z)} = \arg \max_{x\in \{x_k^{*(z)}\} } P_k \ncalX_{kx}^{(z)} \|x-z\|^{-\alpha}.
\end{align}

Owing to the displacement theorem for PPPs~\cite{BacBlaB2009}, any general distribution of $\ncalX_{k}$ can be handled in the downlink analysis of a typical user as long as $\nbbE\left[\ncalX_{k}^{\frac{2}{\alpha}} \right] < \infty$. This is formally discussed in detail in~\cite{KeeBlaC2013,DhiAndJ2013}. The most common assumption for large scale shadowing distribution is lognormal, where $\ncalX_k = 10^{\frac{X_k}{10}}$ such that $X_{k} \sim \ncalN(m_k, \sigma_k^2)$, where $m_k$ and $\sigma_k$ are respectively the mean and standard deviation in dB of the shadowing channel power. For lognormal distribution, $\nbbE\left[\ncalX_k^{\frac{2}{\alpha}} \right] = \exp \left(\frac{\ln 10}{5} \frac{m_k}{\alpha} + \frac{1}{2} \left(\frac{\ln 10}{5} \frac{\sigma_k}{\alpha} \right)^2 \right)$, which can be easily derived using moment generating function (MGF) of Gaussian distribution~\cite{DhiAndJ2013}. The fractional moment is clearly finite if both the mean and standard deviation of the normal random variable $X_{k}$ are finite. For this system model, we now study the availabilities of different classes of BSs.

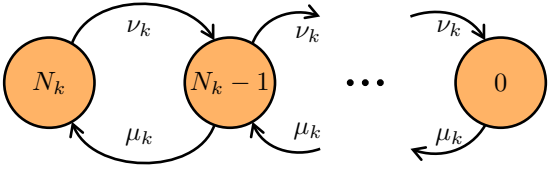
\begin{figure}[t]
\centering
\begin{tikzpicture}[scale=.8,>=angle 60, thick]
\pgfsetlinewidth{1}	
\coordinate (A1) at (0,0);
\coordinate (A2) at (3,0);
\coordinate (A3) at (7.5,0);
\coordinate (B) at (5.25,0);		

\draw [fill=orange!60!] (A1) circle (.75); \node at (A1) {$N_k$};
\draw [fill=orange!60!] (A2) circle (.75); \node at (A2) {$N_k-1$};
\draw [fill=orange!60!] (A3) circle (.75); \node at (A3) {$0$};
\filldraw [black] ($(B)-(.25,0)$) circle (1pt);
\filldraw [black] (B) circle (1pt);
\filldraw [black] ($(B)+(.25,0)$) circle (1pt);

\draw [->] ($(A1)+cos(60)*(.75,0)+sin(60)*(0,.75)$) to [out=70,in=120] ($(A2)+cos(110)*(.75,0)+sin(110)*(0,.75)$); \node at ($(A1)+(1.5,0.9)$) {$\nu_{k}$};

\draw [->] ($(A2)+cos(-110)*(.75,0)+sin(-110)*(0,.75)$) to [out=-110,in=-70] ($(A1)+cos(-60)*(.75,0)+sin(-60)*(0,.75)$); \node at ($(A1)+(1.5,-0.9)$) {$\mu_{k}$};

\draw [->] ($(A2)+cos(60)*(.75,0)+sin(60)*(0,.75)$) to [out=70,in=160] ($(A2)+(1.5,1.1)$); \node at ($(A2)+(1.3,.8)$) {$\nu_{k}$};

\draw [<-] ($(A2)+cos(-60)*(.75,0)+sin(-60)*(0,.75)$) to [out=-70,in=200] ($(A2)+(1.5,-1.1)$); \node at ($(A2)+(1.3,-.8)$) {$\mu_{k}$};

\draw [->] ($(A3)+(-1.5,+1.1)$) to [out=20,in=120] ($(A3)+cos(110)*(.75,0)+sin(110)*(0,.75)$); \node at ($(A3)+(-.85,0.9)$) {$\nu_{k}$};

\draw [<-] ($(A3)+(-1.5,-1.1)$) to [out=-20,in=-120] ($(A3)+cos(-110)*(.75,0)+sin(-110)*(0,.75)$); \node at ($(A3)+(-.85,-0.9)$) {$\mu_{k}$};
\end{tikzpicture}
\caption{Birth-death process modeling the temporal dynamics of the energy available at a $k^{th}$ tier BS.}
\label{fig:BirthDeath_fig}
\end{figure}

\section{Availability Analysis} \label{sec:availability}

The first challenge in studying the model introduced in the previous section lies in characterizing how the energy available at the BS changes over time. Without loss of generality, we index the energy states of a $k^{th}$ tier BS as $0, 1, \ldots, N_k$, and model the temporal dynamics as a continuous time Markov chain (CTMC), in particular a birth-death process, as shown in Fig.~\ref{fig:BirthDeath_fig}. When the BS is ON, the energy increases according to the energy harvesting rate and decreases at a rate that depends upon the number of users served by that BS. When the BS is OFF, it does not serve any users and hence the birth-death process reduces to a birth-only process. We now derive a closed form expression for the rate $\nu_k$ at which the energy is utilized at a typical $k^{th}$ tier BS.

\subsection{Modeling Energy Utilization Rate}
Before modeling the energy utilization rate, there are two noteworthy points. First, if a BS is not available, the load originating from its original area of coverage is directed to the nearby BSs that are available, thus increasing their effective load. Equivalently, the coverage areas of the BSs that are available get expanded to cover for the BSs that are not available, as shown in Fig.~\ref{fig:coverageregions}. The second one is related to the control channel coverage and given in the following remark. Recall that control channel coverage $\pc$ is the probability that the received signal-to-interference-ratio ($\sir$) is greater than the predefined minimum $\sir$ needed to establish a connection with the BS. Thus the users that are not in control channel coverage cannot enter the network and hence cannot access the data channels. Therefore, these users do not account for any additional energy expenditure at the BS.

\begin{remark}[Control channel coverage] \label{rem:Pc}
The control channel coverage $\pc$ is independent of the densities of the BSs in an interference-limited network when the target SIR is the same for all tiers~\cite{JoSanJ2012,DhiGanJ2012,DhiAndJ2013}. While this result will be familiar to those exposed to recent coverage probability analysis using stochastic geometry, it is not directly required in this section except the interpretation that the density of users effectively served by the network is independent of the effective densities of the BSs and hence independent of $\{\rho_k\}$. We will validate this claim in Section~\ref{sec:performance}.
\end{remark}

\begin{figure}
\centering
\includegraphics[width=.8\columnwidth]{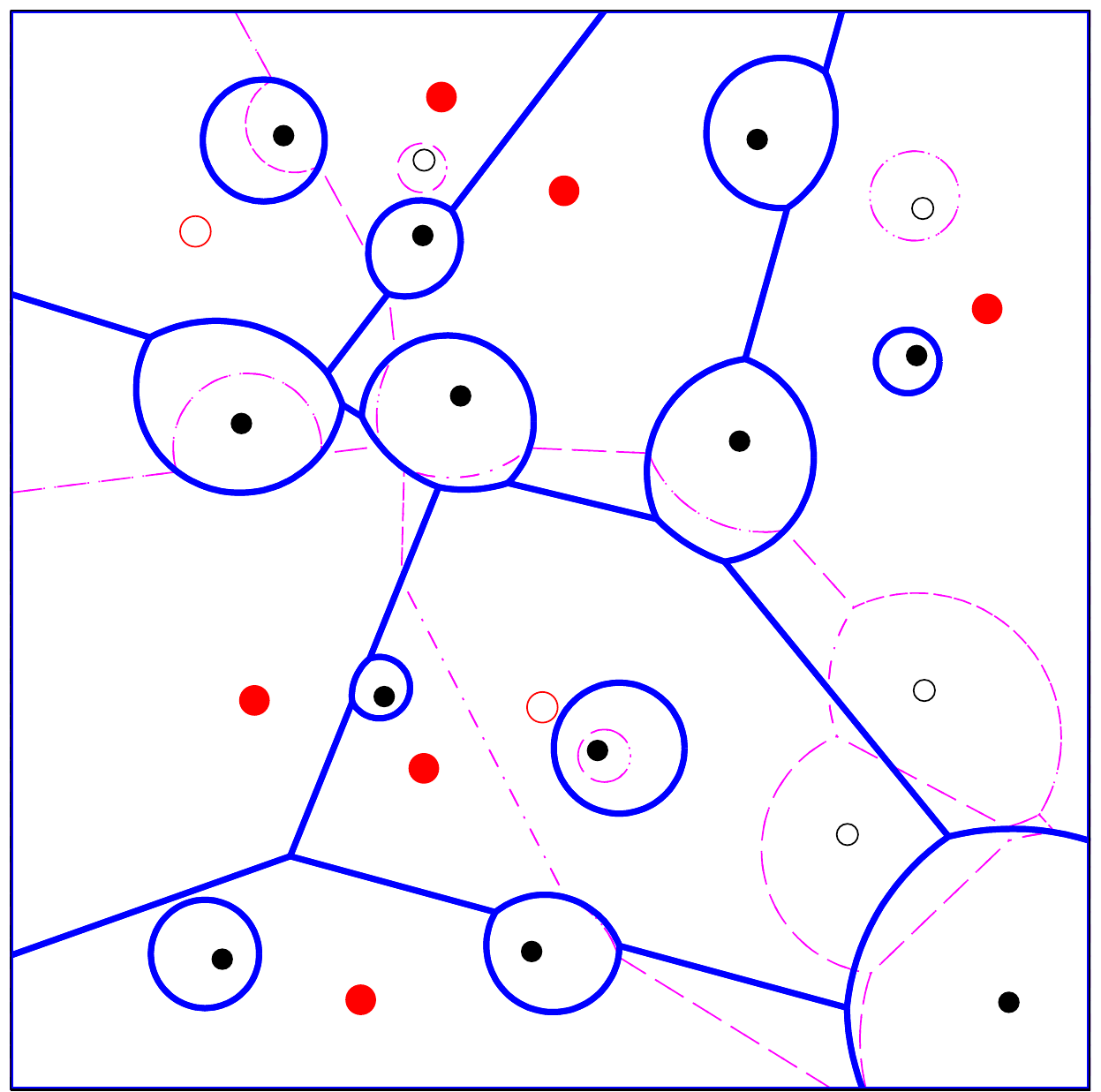}
\caption{Coverage regions for a two-tier energy harvesting cellular network (averaged over shadowing). The unavailable BSs are denoted by hollow circles. The thin lines form coverage regions for the baseline case assuming all the BSs were available.}
\label{fig:coverageregions}
\end{figure}

Assuming fixed energy expenditure for control signaling, only the users that are in control channel coverage will result in additional energy expenditure at the BS. As remarked above, the density of such users is $\pc \lambda_u$. Each user is assumed to require 1 joule of energy per sec such that the net energy utilization process at each BS can be modeled as a Poisson process with mean defined by the average number of users it serves. It should be noted that the assumption of 1 joule energy requirement is without any loss of generality and is made to simplify the notation. To find the average number of users served by a typical BS of each class, we first need to define its service region whose statistics such as its area will, in general, be different for different classes of BSs due to the differences in the transmit powers as evident from Fig.~\ref{fig:coverageregions}. The service region can be formally defined as follows.
\begin{definition}[Service region]
The service region $\calA_k(x_k) \subset \nbbR^2$ of the $k^{th}$-tier BS located at $x_k \in \PPPak$ is $\calA_k(x_k) =$
\begin{align}
\left\{ z \in \R^2:  x_k = \arg \max_{x \in \{x_j^{*(z)}\}} P_j \ncalX_j^{(z)} \|x - z\|^{-\alpha}, \right. \nonumber \\
{\rm where}\ \left. x_j^{*(z)} =  \arg \max_{x \in \PPPaj} P_j \ncalX_j^{(z)} \|x-z\|^{-\alpha} \right\}.
\end{align}
\end{definition}
We now derive the average area of the service region of a typical BS of each tier in the following Lemma.

\begin{lemma}[Average area of the service region] \label{lem:avgservicearea}
The average area of the service region of a $k^{th}$ tier typical BS is given by
\begin{align}
\E[|\calA_k|] = \frac{ \nbbE \left[\ncalX_k^{\frac{2}{\alpha}} \right] P_k^{\frac{2}{\alpha}}}{\sum_{j =1}^K  \rho_j \lambda_j \nbbE \left[\ncalX_j^{\frac{2}{\alpha}} \right] P_j^{\frac{2}{\alpha}}}.
\end{align}
\end{lemma}

\begin{proof}
The proof follows from the definition of the service area using basic ideas from Palm calculus and is given in Appendix~\ref{Appendix:avgservicearea}.
\end{proof}

Using the expression for average area, the average number of users served by a typical BS of $k^{th}$ tier, equivalently the energy utilization rate, is now given by the following corollary.
\begin{cor}[Energy utilization rate] \label{cor:energyutilizationrate}
The energy utilization rate, i.e., the number of units of energy required per second, at a typical BS of $k^{th}$ tier is given by
\begin{align}
\nu_k = \pc \lambda_u \E[|\calA_k|] = \frac{\pc \lambda_u \nbbE \left[\ncalX_k^{\frac{2}{\alpha}} \right] P_k^{\frac{2}{\alpha}}}{\sum_{j =1}^K  \rho_j \lambda_j \nbbE \left[\ncalX_j^{\frac{2}{\alpha}} \right] P_j^{\frac{2}{\alpha}}}, \label{eq:energyutilizationrate}
\end{align}
where recall that $\pc$ denotes the coverage probability, which is independent of the availabilities and will be calculated later in this section and is given by~\eqref{eq:pc_int_limited}.
\end{cor}
\begin{remark}[Invariance to shadowing distribution] \label{rem:shadowinvar_EnUtil}
From \eqref{eq:energyutilizationrate}, note that the energy utilization rate $\nu_k$ is invariant to the shadowing distribution of all the tiers if $\nbbE \left[ \ncalX_j^{\frac{2}{\alpha}} \right] = \nbbE \left[ \ncalX_k^{\frac{2}{\alpha}} \right]$, for all $j, k \in \ncalK$. For lognormal shadowing, this corresponds to the case when $m_j = m_k$ and $\sigma_j = \sigma_k$, for all $j, k \in \ncalK$.
\end{remark}
It should be noted that the availabilities of various tiers are still unknown and even if all the system parameters are given, it is still not possible to determine the energy utilization rate from the above expression. This will lead to fixed point expressions in terms of availabilities, which is the main focus of the rest of this section. It is also worth mentioning that the energy utilization rate derived above is just for the service of the active users. There are some other components of energy usage, e.g., control channel signaling and backhaul that are not modeled. While we can incorporate their effect in the current model by assuming fixed energy expenditure and deducting it directly from the energy arrival rate, a more formal treatment of these components is left for future work.

\subsection{Availabilities for a Simple Operational Strategy}
After deriving the energy utilization rate in Corollary~\ref{cor:energyutilizationrate} and recalling that the energy harvesting rate is $\mu_k$, we can, in principle, derive BS availabilities for a variety of uncoordinated operational strategies. We begin by looking at a very simple strategy in which a BS is said to be available when it is not in energy state $0$, i.e., it has at least one unit of energy. As shown later in this section, this strategy is of fundamental importance in characterizing the {\em availability region} for the set of general uncoordinated strategies. The availability of a $k^{th}$ tier BS under this strategy can be derived directly from the stationary distribution of the birth-death process as
\begin{align}
\rho_k &= 1 - \left(\frac{1-\frac{\mu_k}{\nu_k}}{1-\left(\frac{\mu_k}{\nu_k} \right)^{N_k+1}} \right) \\
&=1 - \left(\frac{1 - \frac{\mu_k \sum_{j =1}^K  \rho_j \lambda_j \nbbE \left[\ncalX_j^{\frac{2}{\alpha}} \right] P_j^{\frac{2}{\alpha}}}{\pc \lambda_u \nbbE \left[\ncalX_k^{\frac{2}{\alpha}} \right] P_k^{\frac{2}{\alpha}}} } {1 - \left(\frac{\mu_k \sum_{j =1}^K  \rho_j \lambda_j \nbbE \left[\ncalX_j^{\frac{2}{\alpha}} \right] P_j^{\frac{2}{\alpha}} } {\pc \lambda_u \nbbE \left[\ncalX_k^{\frac{2}{\alpha}} \right] P_k^{\frac{2}{\alpha}}}\right)^{N_k+1}} \right).
\label{eq:fp_k}
\end{align}

Interestingly we get a set of $K$ fixed point equations in terms of availabilities, one for each tier. Clearly $\rho_k \equiv 0$, $\forall\ k \in \ncalK$, is a trivial solution for this set of fixed point equations. However, this means that none of the BSs is available for service, which physically means that the users are in ``outage'' if there is no other, in particular positive, solution for the set of fixed point equations. We will formalize this notion of outage, resulting from energy unavailability, later in this section. Due to the form of these equations, it is not possible to derive closed form expressions for the positive solution(s) of $\{\rho_k\}$. However, it is possible to establish a necessary and sufficient condition for the existence and uniqueness of a non-trivial positive solution. Before establishing this result, we show that the function of $\{\rho_k\}$ on the right hand side of \eqref{eq:fp_k} satisfies certain key properties. For notational simplicity, we call this function corresponding to $k^{th}$ tier as $g_k: \nbbR^K \rightarrow \nbbR$, using which the set of fixed point equations given by \eqref{eq:fp_k} can be expressed in vector form as
\begin{align}
\left[\begin{array}{c} \rho_1 \\ \rho_2 \\ \vdots \\ \rho_K \end{array} \right] =  \left[\begin{array}{c} g_1(\rho_1, \rho_2, \ldots, \rho_K) \\ g_2(\rho_1, \rho_2, \ldots, \rho_K) \\ \vdots \\ g_K(\rho_1, \rho_2, \ldots, \rho_K) \end{array} \right] = \Xi(\rho_1, \rho_2, \ldots, \rho_K),
\label{eq:fp_vec}
\end{align}
where we further define function $\Xi:\nbbR^K \rightarrow \nbbR^K$ for simplicity of notation. Our first goal is to study the properties of function $g_k: \nbbR^K \rightarrow \nbbR$, which can be rewritten as
\begin{align}
g_k (\nbx) = 1 - \left(\frac {1 - \sum_{j=1}^K a_j x_j}   {1 - \left( \sum_{j=1}^K a_j x_j \right)^{N}} \right),
\label{eq:func_fp}
\end{align}
where $\nbx \in \nbbR^K$, $N>1$, and $a_k \in \nbbR_+$ for all $k\in \ncalK$. The relevant properties are summarized in the following Lemma.

\begin{lemma}[Properties] \label{lem:properties}
The function $g_k(\nbx): \nbbR^K \rightarrow \nbbR$ defined by \eqref{eq:func_fp} satisfies following properties for all $a_k > 0$, $k \in \ncalK$:
\begin{enumerate}
\item $g_k(\nbx)$ is an element-wise increasing function of $\nbx$.
\item $g_k(\nbx)$ is concave, i.e., it is a concave function of $x_k \in \nbbR$ for all $k \in \ncalK$.
\end{enumerate}
\end{lemma}
\begin{proof}
The proof is given in Appendix~\ref{Appendix:properties}.
\end{proof}

Lemma~\ref{lem:properties} can be easily extended to the function $\Xi:\nbbR^K \rightarrow \nbbR^K$ to show that it is also a monotonically increasing and concave function. The conditions for existence and uniqueness of the fixed point for such functions can be characterized by specializing Tarski's theorem~\cite{TarJ1955} for concave functions. The result is stated below. To the best of the knowledge of the authors, it first appeared in~\cite[Theorem 3]{KenJ2001}. Since the proof is given in~\cite{KenJ2001}, it is skipped here.

\begin{theorem}[Fixed point for increasing concave functions] \label{thm:fp_concave}
Suppose $\Xi: \nbbR^n \rightarrow \nbbR^n$ is an increasing and strictly concave function satisfying the following two properties:
\begin{enumerate} 
\item $\Xi(0) \geq 0$, $\Xi(a) > a$ for some $a \in \nbbR^n_+$, 
\item $\Xi(b) < b$ for some $b>a$. 
\end{enumerate}
Then $\Xi$ has a unique positive fixed point.
\end{theorem}

Before deriving the main result about the existence and uniqueness of positive solution for the set of fixed point equations~\eqref{eq:fp_k}, for cleaner exposition we state the following intermediate result that establishes equivalence between an energy conservation principle and a key set of conditions.

\begin{lemma}[Equivalence] \label{lem:equivalence}
For $\rho_k > 0, \forall k$, the following sets of conditions are equivalent, i.e., \eqref{eq:sol_existence} $\Leftrightarrow$ \eqref{eq:sol_condition}
\begin{align}
\frac{\mu_k \sum\limits_{j =1}^K  \rho_j \lambda_j \nbbE \left[\ncalX_j^{\frac{2}{\alpha}} \right] P_j^{\frac{2}{\alpha}}}{\rho_k \pc \lambda_u \nbbE \left[\ncalX_k^{\frac{2}{\alpha}} \right] P_k^{\frac{2}{\alpha}}} &> 1, \forall k \in \ncalK \label{eq:sol_existence} \\
\sum_{k=1}^K \lambda_k \mu_k &> \lambda_u \pc, \label{eq:sol_condition}
\end{align}
where \eqref{eq:sol_condition} is simply the energy conservation principle, i.e., the net energy harvested by all the tiers should be greater than the effective energy required by all the users.
\end{lemma}

\begin{proof}
The proof is given in Appendix~\ref{Appendix:equivalence}.
\end{proof}

Using Theorem~\ref{thm:fp_concave} and Lemma~\ref{lem:equivalence}, we now derive the main result of this subsection.

\begin{theorem} \label{thm:energyconservation}
The necessary and sufficient condition for the existence of a positive solution $\rho_k > 0$, $\forall\ k \in \ncalK$ for the system of fixed point equations given by~\eqref{eq:fp_k} is
\begin{align}
\sum_{k=1}^K \lambda_k \mu_k > \lambda_u \pc.
\label{eq:sol_condition_thm}
\end{align}
\end{theorem}

\begin{proof}
For sufficiency, it is enough to show that the given condition is sufficient for the function $\Xi: \nbbR^K \rightarrow \nbbR^K$ defined by~\eqref{eq:fp_vec} to satisfy both the properties listed in Theorem~\ref{thm:fp_concave}. Further, it is enough to show this for each element function $g_k: \nbbR^K \rightarrow \nbbR$ of $\Xi$. For $\rho_k \neq 0$, the function $g_k$, as a function of $\rho_k$ can be expressed as
\begin{align}
g_k(\rho_k) = 1 - \left( \frac{1-\kappa_k \rho_k}{ 1 - (\kappa_k \rho_k)^{N_k+1}}\right),
\label{eq:fp_k1}
\end{align}
where
\begin{align}
\kappa_k = \frac{\mu_k \sum_{j =1}^K  \rho_j \lambda_j \nbbE \left[\ncalX_j^{\frac{2}{\alpha}} \right] P_j^{\frac{2}{\alpha}}}{\rho_k \pc \lambda_u \nbbE \left[\ncalX_k^{\frac{2}{\alpha}} \right] P_k^{\frac{2}{\alpha}}}.
\end{align}
Note that the function $g_k(\rho_k) < 1$ for finite $N_k$. Now setting $b$, as defined in Theorem~\ref{thm:fp_concave}, equal to $1$, it is enough to find conditions under which $\exists\ a < b$ such that $g_k(a) > a$. Since $g_k(\rho_k) = 0$ for $\rho_k \rightarrow 0$, for the existence of $a$ such that $g_k(a) > a$ it is enough to show that $g'(\rho_k) > 1$ for $\rho_k \rightarrow 0$. 
Furthermore, it is easy to show that $g'(\rho_k) = \kappa_k$ for $\rho_k \rightarrow 0$, which leads to the condition $\kappa_k > 1$ for the existence of $a$ as defined above. This leads to the following set of inequalities for $1\leq k \leq K$:
\begin{align}
\frac{\mu_k \sum_{j =1}^K  \rho_j \lambda_j \nbbE \left[\ncalX_j^{\frac{2}{\alpha}} \right] P_j^{\frac{2}{\alpha}}}{\rho_k \pc \lambda_u \nbbE \left[\ncalX_k^{\frac{2}{\alpha}} \right] P_k^{\frac{2}{\alpha}}} > 1.
\end{align}
From Lemma~\ref{lem:equivalence}, this set of conditions is the same as \eqref{eq:sol_condition_thm} and hence proves that \eqref{eq:sol_condition_thm} is a sufficient condition for the existence and uniqueness of the positive solution for $\{\rho_k\}$. 

To show that the given condition is also necessary, we construct a simple counter example. Let $K=1$ and drop all the subscripts denoting the indices of tiers for notational simplicity. The fixed point equation for this simple setup is
\begin{align}
\rho = 1 -\left(\frac{1 - \frac{\mu \rho \lambda}{\pc \lambda_u}}    {1 - \left( \frac{\mu \rho \lambda}{\pc \lambda_u} \right)^N} \right) = g(\rho),
\end{align}
It is easy to show that $g(\rho)$ does not have a positive fixed point when $\mu \lambda < \pc \lambda_u$, which proves that the given condition~\eqref{eq:sol_condition_thm} is also necessary.
\end{proof}
The existence and uniqueness of the positive solution for the BS availabilities $\{\rho_k\}$ will play a crucial role in establishing the availability region later in this section. The unique positive solution for $\{\rho_k\}$ can be computed easily using fixed-point iteration. Before concluding this section, it is important to formalize some key ideas.

\begin{remark}[Energy outage]
From Theorem~\ref{thm:energyconservation}, it is clear that the total energy harvested by the HetNet must be greater than the total energy demand to guarantee a positive solution for the availabilities $\{\rho_k\}$. However, if this condition is not met, the system may drop a certain fraction of users to ensure that the resulting density of users $\lambda'_u$ is such that $\sum_{k=1}^K \lambda_k \mu_k > \lambda_u' \pc$. The rest of the users are said to be in outage due to energy unavailability, or in short ``energy outage''. The probability of a user being in energy outage is
\begin{align}
\ncalO_e = 1 - \frac{\lambda'_u}{\lambda_u} \geq 1 - \frac{\sum_{k=1}^K \lambda_k \mu_k}{\lambda_u \pc},
\end{align}
where the lower bound is strictly positive if $\sum_{k=1}^K \lambda_k \mu_k < \lambda_u \pc$, i.e., when condition~\eqref{eq:sol_condition_thm} is not met.  However, if the condition~\eqref{eq:sol_condition_thm} is met, it is in principle possible to make $\ncalO_e = 0$. The exact characterization of energy outage will depend upon the protocol design and is out of the scope of this paper. In the rest of the paper, we will assume that the condition~\eqref{eq:sol_condition_thm} is met and hence $\ncalO_e = 0$.
\end{remark}

\begin{remark}[Effect of battery capacity on availability]
Note that the function $g_k$ is an increasing function of $N_k$ from which it directly follows that the availability of a particular class of BSs increases with the increase in the battery capacity.
\end{remark}
\begin{remark}[Effect of availabilities of other tiers on $\rho_k$] \label{rem:effectothers}
From Lemma~\ref{lem:properties}, it follows that $g_k$ is an increasing function of not only $\rho_k$ but also of $\rho_j$, $j \neq k$. This implies that the availability of $k^{th}$ tier increases if the availability of one or more of the other tiers is increased. This is consistent with the intuition that if the availability of any tier is increased, the effective load on other tiers decreases hence increasing their availabilities.
\end{remark}

\begin{definition}[Over-provisioning factor]
As mentioned above, we will henceforth assume that the system is over provisioned in terms of energy harvesting, i.e., $\sum_{k=1}^K \lambda_k \mu_k > \lambda_u \pc$. For cleaner exposition, it is useful to define an over-provisioning factor $\gamma$ as the ratio of  total energy harvested in the network and the effective energy demand, i.e.,
\begin{align}
\gamma = \frac{\sum_{k=1}^K \lambda_k \mu_k}{\lambda_u \pc} > 1.
\end{align}
\end{definition}

So far we focused on a particular strategy, where a BS is said to be available if it is not in the $0$ energy state, i.e., it has at least one unit of energy. In the next subsection, we develop tools to study availabilities for any general uncoordinated strategy using stopping/hitting time analysis. Our analysis will concretely demonstrate that the simple strategy discussed above maximizes the BS availabilities over the space of general uncoordinated strategies. Extending these results further, we will characterize the availability region that is achievable by the set of general uncoordinated strategies. 

\begin{figure}
\centering
\includegraphics[width=\columnwidth]{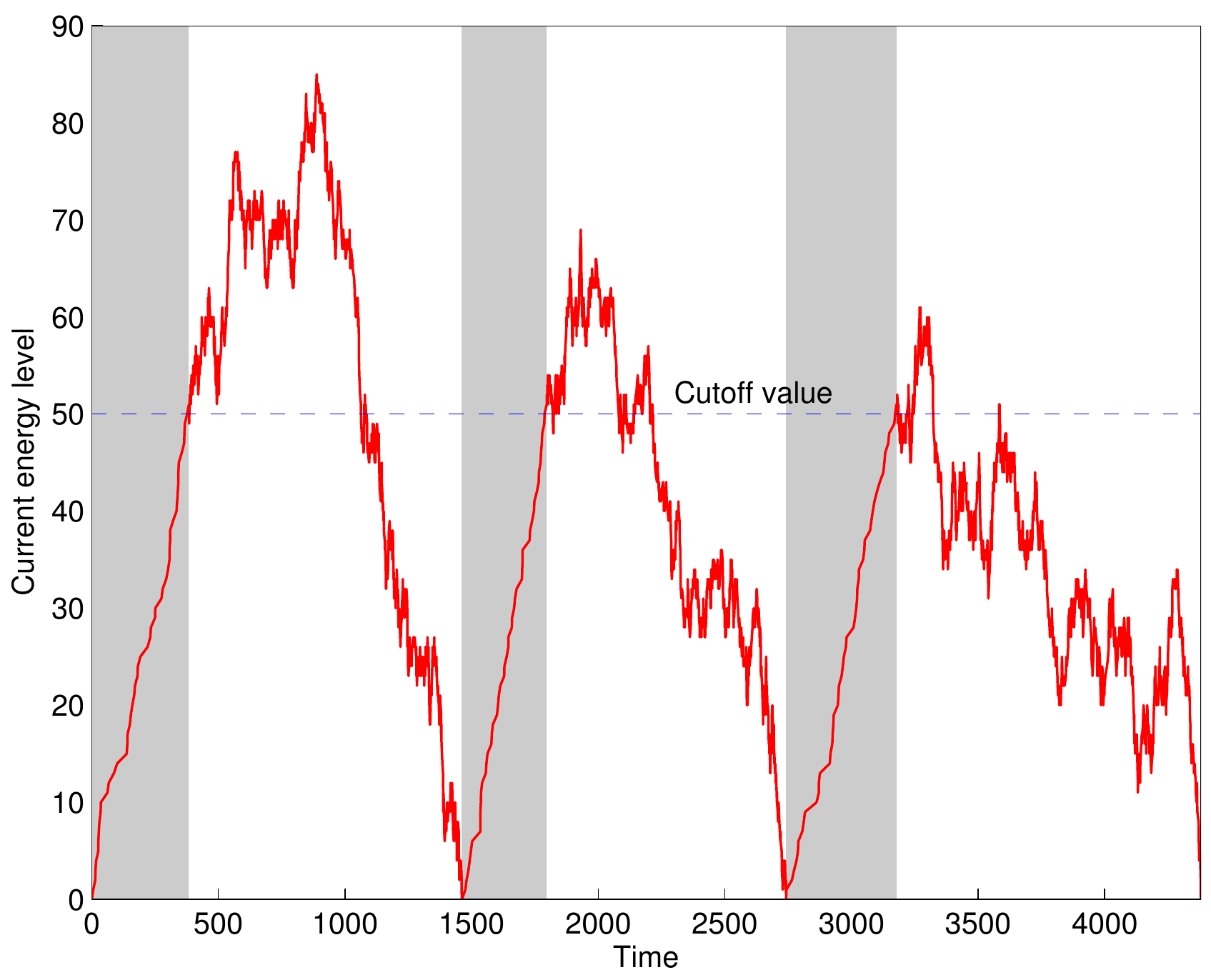}
\caption{Illustration of how the energy level changes over time. The time for which BS is in OFF state is shaded. The unit of time is irrelevant.}
\label{fig:En_vs_Time}
\end{figure}

\subsection{Availabilities for any General Uncoordinated Strategy}
We first focus on a general set of strategies $\{\ncalS_k(N_{k\min}, N_{kc})\}$ in which a BS toggles its state based solely on its current energy level, i.e., a $k^{th}$ tier BS toggles to OFF state when its energy level reaches some level $N_{k\min}$ and toggles back to ON state when the energy level reaches some predefined cutoff value $N_{kc} > N_{k\min}$ as shown in Fig.~\ref{fig:En_vs_Time}. Although not required for this analysis, it should be noted that the cutoff value $N_{kc}$ can be changed by the network if necessary on an even larger time scale than the time scale over which the BSs are turned ON/OFF. Now note that for the proposed model, the strategies $\{\ncalS_k(N_{k\min}, N_{kc})\}$ with energy storage capacity $N_k$ and $\{\ncalS_k(0, N_{kc}-N_{k\min})\}$ with energy storage capacity $N_k - N_{k\min}$, are equivalent because when the BS is turned OFF at a non-zero energy level $N_{k\min}$ in the first set of strategies, it effectively reduces the energy storage capacity to $N_k - N_{k\min}$. Therefore, without any loss of generality we fix $N_{k\min}=0$ (for all tiers) and denote this strategy by $\ncalS_k(N_{kc})$ for notational simplicity. For this strategy, we denote the time for which a $k^{th}$ tier BS is in the ON state after it toggles from the OFF state by $J_{k_1}(N_{kc})$ and the time for which it remains in the OFF state after toggling from the ON state by $J_{k_2}(N_{kc})$. The cutoff value in the arguments will be dropped for notational simplicity wherever appropriate. The cycles of ON and OFF times go on as shown in Fig.~\ref{fig:En_vs_Time}. It is worth highlighting that both $J_{k_1}$ and $J_{k_2}$ are in general random variables due to the randomness involved in both the energy availability and its utilization, e.g., $J_{k_1}$ can be formally expressed as
$J_{k_1}(N_{kc}) = \inf \{ t: \ncalE_k(t) = 0 | \ncalE_k(0)= N_{kc} \}$,
where $\ncalE_k(t)$ denotes the energy level of a $k^{th}$ tier BS at time $t$. For this setup, the availabilities depend only on the means of $J_{k_1}$ and $J_{k_2}$ as shown in the following Lemma.

\begin{lemma}[Availability]
The availability of a $k^{th}$ tier BS for any operational strategy can be expressed as
\begin{align}
\rho_k = \frac{\E[J_{k_1}]}{\E[J_{k_1}] + \E[J_{k_2}]} = \frac{1}{1+ \frac{\E[J_{k_2}]}{\E[J_{k_1}]}},
\end{align}
where $\E[J_{k_1}]$ is the mean time a BS spends in the ON state and $\E[J_{k_2}]$ is the mean time it spends in the OFF state.
\end{lemma}
\begin{proof}
For a particular realization, let $\{J_{k_1}^{(i)}\}$ and $\{J_{k_2}^{(i)}\}$ be the sequences of ON and OFF times, respectively, with $i$ being the index of the ON-OFF cycle. The availability can now be expressed as the fraction of time a BS spends in the ON state, which leads to
\begin{align}
\rho_k = \lim_{n \rightarrow \infty} \frac{\sum_{i=1}^n J_{k_1}^{(i)}}{\sum_{i=1}^n J_{k_1}^{(i)} + \sum_{i=1}^n J_{k_2}^{(i)}}.
\end{align}
The proof follows by dividing both the numerator and the denominator by $n$ and invoking the law of large numbers.
\end{proof}
To set up a fixed point equation similar to \eqref{eq:fp_k} for the strategy $\ncalS_k(N_{kc})$, we need closed form expressions for the mean ON time $\E[J_{k_1}]$ and the mean OFF time $\E[J_{k_2}]$. Note that the OFF time for $\ncalS_k(N_{kc})$ is simply the time required to harvest $N_{kc}$ units of energy, which is the sum of $N_{kc}$ exponentially distributed random variables, each with mean $1/\mu_k$. Therefore, 
\begin{align}
\E[J_{k_2}] = \frac{N_{kc}}{\mu_k} \Rightarrow \rho_k = \frac{1}{1+ \frac{N_{kc}}{\mu_k \E[J_{k_1}]}}.
\label{eq:rho_k_intON}
\end{align} 
To derive $\E[J_{k_1}]$, we first define the generator matrix for the birth-death process corresponding to a $k^{th}$ tier BS as $A_k =$
\begin{align}
\left[  \begin{array}{cccccc}
-\mu_k & \mu_k & 0 & \cdots & 0&0\\
\nu_k & -\mu_k - \nu_k & \mu_k & \cdots & 0&0\\
0 & \nu_k & -\mu_k - \nu_k & \cdots & 0&0 \\
\vdots & \vdots & & \ddots & &\\
0 & 0 & 0 & \cdots & \nu_k & -\nu_k
\end{array}\right],
\label{eq:A_k}
\end{align}
where the states are ordered in the ascending order of the energy levels, i.e., the first column corresponds to the energy level $0$. To complete the derivation, we need the following technical result. Please refer to Proposition 5.7.2 of~\cite{ResB2005} for a more general version of this result and its proof.
\begin{lemma}[Mean hitting time] \label{lem:MeanHitTime}
If the embedded discrete Markov chain of the CTMC is irreducible then the mean time to hit energy level $0$ (state $1$) starting from energy level $i$ (state $i+1$) is
\begin{align}
\nbbE[J_{k_1}(i)] = \left(\left(-B_k \right)^{-1} \mathbbm{1} \right) (i),
\label{eq:meanONgen}
\end{align}
where $\mathbbm{1}$ is a column vector of all $1$s, and $B_k$ is a $(N_k-1) \times (N_k-1)$ sub-matrix of $A_k$ obtained by deleting first row and column of $A_k$.
\end{lemma}
For $A_k$ given by \eqref{eq:A_k}, we can derive a closed form expression for each element of $\left(-B_k \right)^{-1}$ 
after some algebraic manipulations. The $(i,j)^{th}$ element can be expressed as
\begin{align}
\left(-B_k \right)^{-1} (i,j) = \frac{1}{\nu_k^j} \sum_{n=1}^{\min(i,j)} \mu_k^{j-n} \nu_k^{n-1}.
\label{invB_k_i_j}
\end{align}
Now substituting \eqref{invB_k_i_j} back in \eqref{eq:meanONgen} gives us the mean ON time for any strategy $\ncalS_k(N_{kc})$, which when substituted in \eqref{eq:rho_k_intON} gives a fixed point equation in $\{\rho_k\}$ similar to \eqref{eq:fp_k}, as illustrated below for the two policies of interest.

\subsubsection{Policy 1 ($\ncalS_k(1)$)} In this policy, each BS serves users until it depletes all its energy after which it toggles to OFF state. It toggles back to ON state after it has harvested one unit of energy. Using \eqref{eq:meanONgen} and \eqref{invB_k_i_j}, the mean ON time $\nbbE[J_{k1}]$ for this policy can be expressed as
\begin{align}
\nbbE[J_{k_1}] = \frac{1}{\nu_k} \frac{1-\left(\frac{\mu_k}{\nu_k}\right)^{N_k}}{1-\left(\frac{\mu_k}{\nu_k}\right)},
\end{align}
which when substituted into \eqref{eq:rho_k_intON} leads to
\begin{align}
\rho_k = 1 - \frac{1-\frac{\mu_k}{\nu_k}}{ 1-\left(\frac{\mu_k}{\nu_k}\right)^{N_k+1}   },
\end{align}
which is the same fixed point equation as \eqref{eq:fp_k}. This establishes an equivalence between this policy and the one studied in the previous subsection.  In particular, this policy is an achievable strategy to achieve the same availabilities as the ones possible with the strategy studied in the previous subsection.

\subsubsection{Policy 2 ($\ncalS_k(N_k)$)} As in the above policy, each BS serves users until it depletes all its energy after which it toggles to OFF state. Under this policy, the BS waits in the OFF state until it harvests $N_k$ units of energy, i.e., its energy storage module is completely charged. Using \eqref{eq:meanONgen} and \eqref{invB_k_i_j}, $\nbbE[J_{k_1}]$ can be expressed as
\begin{align}
\nbbE[J_{k_1}] = \frac{1}{\mu_k - \nu_k} \frac{\mu_k}{\nu_k} \frac{1-\left(\frac{\mu_k}{\nu_k}\right)^{N_k}}{1-\left(\frac{\mu_k}{\nu_k}\right)} - \frac{N_k}{\mu_k - \nu_k},
\end{align}
which can be substituted in \eqref{eq:rho_k_intON} to derive the fixed point equation for this policy.

While policy 1 will be of fundamental importance in establishing the availability region, we will also consider policy 2 at several places to highlight key points. We now prove the following theorem, which establishes a fundamental upper limit on the availabilities of various types of BSs that cannot be surpassed by any uncoordinated strategy. Please note that although we have discussed only ``energy-based'' uncoordinated strategies so far, the general set of uncoordinated strategies also additionally includes timer-based, and the combination of energy and timer-based strategies. This is taken into account in the proof of the following theorem.

\begin{theorem} \label{thm:S1_superior}
For a given $K$ tier network, the availabilities of all the classes of BSs are jointly maximized over the space of general uncoordinated strategies if each tier follows strategy $\ncalS_{k}(1)$. The availabilities are strictly lower if any one or more tiers follow $\ncalS_{k}(i)$, $i>1$, with a non-zero probability.
\end{theorem}

\begin{proof}
From \eqref{eq:rho_k_intON}, note that the availability for a $k^{th}$ tier BS is maximized if $\nbbE[J_{k_1}(N_{kc})]/N_{kc}$ is maximized. Using \eqref{eq:meanONgen} and \eqref{invB_k_i_j}, it is straightforward to show that
\begin{align}
\arg \max_{1\leq i \leq N_k} \frac{\nbbE[J_{k_1}(i)]}{i} = 1. 
\end{align}
The proof now follows from the fact that if any tier follows strategy $\ncalS_{k}(i)$ ($i>1$) with a non-zero probability, its availability will be strictly lower than that of $\ncalS_{k}(1)$, which increases the effective load on other tiers and hence decreases their availabilities, as discussed in Remark~\ref{rem:effectothers}. Therefore, to jointly maximize the availabilities of all the tiers, each tier has to follow $\ncalS_k(1)$.

Now note that any strategy that is fully or partly based on a timer can be thought of as an arbitrary combination of $\ncalS_{k} (i)$, where $i>1$ with some non-zero probability. Hence the availabilities for such strategies are strictly lower than $\ncalS_{k}(1)$. 
\end{proof}

Using this result we now characterize the availability region for the set of general uncoordinated strategies.

\subsection{Availability Region}
We begin this subsection by formally defining the availability region as follows.
\begin{definition}[Availability region]
Let $\nfrakR^{\rm (UC)} \subset \nbbR^K$ be the set of availabilities $(\rho_1, \rho_2, \ldots, \rho_K) \in \nbbR^K$ that are achievable by a given uncoordinated strategy $\ncalS^{\rm (UC)}$. The availability region is now defined as
\begin{align}
\nfrakR = \cup_{\ncalS^{\rm (UC)}}   \nfrakR^{\rm (UC)},
\end{align}
where the union is over all possible uncoordinated strategies.
\end{definition}

From Theorem~\ref{thm:S1_superior} we know that the availabilities of all the tiers are jointly maximized if they all follow strategy $\ncalS_k(1)$. For notational ease, we define these maximum availabilities by $\rho^{\max} = (\rho^{\max}_1, \rho^{\max}_2 \ldots \rho^{\max}_K)$. This provides a trivial upper bound on the availability region as follows
\begin{align}
\nfrakR \subseteq \{\rho \in \nbbR^K: \rho_k \leq \rho^{\max}_k, \forall\ k \in \ncalK\},
\label{eq:achievability_ub}
\end{align}
which is simply an orthotope in $\nbbR^K$. Our goal now is to characterize the exact availability region as a function of key system parameters. As a by product, we will show that the upper bound given by \eqref{eq:achievability_ub} is rather loose. For cleaner exposition, we will refer to Fig.~\ref{fig:availability_region}, which depicts the exact availability region for a two-tier setup along with the bound given by~\eqref{eq:achievability_ub}. Before stating the main result, denote by $\rho_k^*(\{\rho_j\} \setminus \rho_k)$ the maximum availability achievable for the $k^{th}$ tier BSs, given the availabilities of the other $K-1$ tiers. It is clearly a function of $(\rho_1, \ldots \rho_{k-1}, \rho_{k+1}, \ldots, \rho_K)$. Following the notation introduced in \eqref{eq:fp_vec}, $\rho_k^*(\{\rho_j\} \setminus \rho_k)$ (denoted by $\rho_k^*$ for notational simplicity) can be expressed as
\begin{align}
\rho_k^* = g_k(\rho_1, \ldots \rho_k^*, \ldots, \rho_K),
\label{eq:rhok_star}
\end{align}
where $\rho_k^*$ is the solution to the fixed point equation given the availabilities of the other $K-1$ tiers. Recall that while defining $\rho_k^*$ in terms of $g_k$, we used Theorem~\ref{thm:S1_superior}, where we proved that strategy $\ncalS_k(1)$ maximizes availability for any given tier and  also leads to the same set of fixed point equations as given by~\eqref{eq:fp_vec}. In Fig.~\ref{fig:availability_region}, the solid line denotes $\rho_2^*(\rho_1)$, and the dotted line denotes $\rho_1^*(\rho_2)$. We remark on the achievability of the availabilities corresponding to these lines for $\rho_k \leq \rho_k^{\max}$, $\forall k \in \ncalK$ below.

\begin{remark}[Achievability of $\rho_1^*(\rho_2)$ and $\rho_1^*(\rho_2)$] \label{rem:achievability}
To show that for $\rho_1 \leq \rho^{\max}_1$, all the points on $\rho_2^*(\rho_1)$ are achievable, consider point $G = (\rho_1^G, \rho_2^G)$ in Fig.~\ref{fig:availability_region}. Given $\rho_2^G$, the maximum possible availability for first tier corresponds to point $H$ on $\rho_1^*(\rho_2^G)$, which further corresponds to strategy $\ncalS(1)$. Clearly $\rho_1^G \leq \rho_1^*(\rho_2^G)$, and hence achievable by some uncoordinated strategy. One option is to time share between $\ncalS(1)$ and a fixed timer that keeps a BS OFF despite having energy to serve users. The timer can be appropriately adjusted such that the effective availability is $\rho_1^G$. 
Likewise, all the points on $\rho_1^*(\rho_2)$ are also achievable. Clearly, this construction easily extends to general $K$ tiers.
\end{remark}

Using these insights, we now derive the exact availability region for the set of uncoordinated strategies in the following theorem.

\begin{figure}
\centering
\includegraphics[width=\columnwidth]{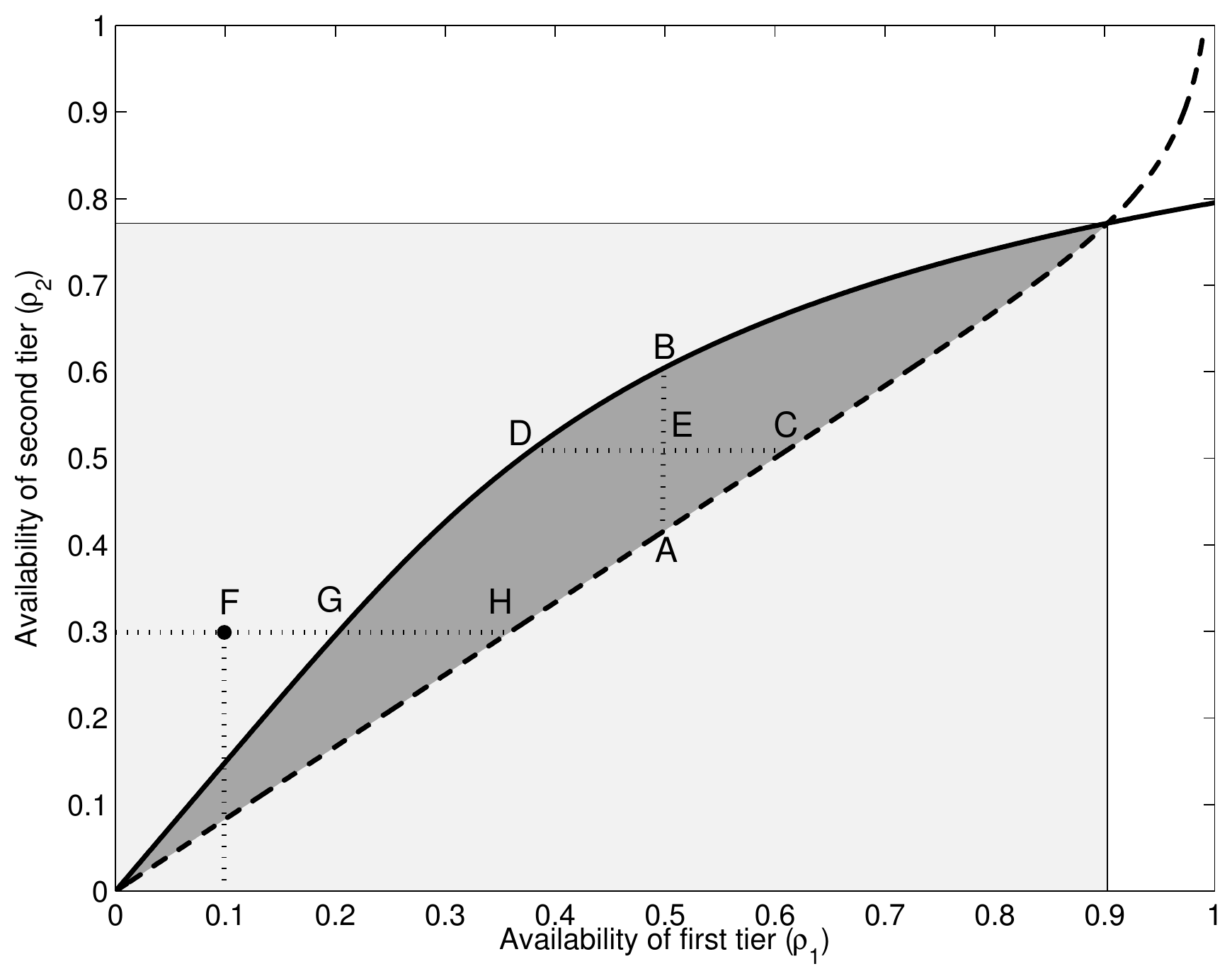}
\caption{Availability region for a two-tier HetNet. The upper bound and the exact availability regions are respectively highlighted in light and dark shades. Setup: $\alpha=4, K=2, N_1 = 10, N_2 = 8, \gamma=1.1, \mu_1=2, \mu_2=1, \lambda_2 = 10\lambda_1$, $m_1 = m_2$, $\sigma_1 = \sigma_2$.}
\label{fig:availability_region}
\end{figure}

\begin{theorem}[Availability region] \label{thm:availability_region}
The availability region for the set of general uncoordinated strategies is
\begin{align}
\nfrakR = \{\rho \in \nbbR^K: \rho_k \leq \rho_k^*(\{\rho_j\} \setminus \rho_k), \forall\ k \in \ncalK\}.
\label{eq:availability_region}
\end{align}
\end{theorem}
  
\begin{proof}
To show that $\nfrakR$ defined by \eqref{eq:availability_region} is in fact the availability region, it is enough to show that $\rho \in \nfrakR$ is achievable and $\rho \notin \nfrakR$ is not achievable. For ease of exposition, we refer to Fig.~\ref{fig:availability_region} and prove for $K=2$, with the understanding that all the arguments trivially extend to general $K$. To show that $\rho \in \nfrakR$ is achievable, consider point $E$ in Fig.~\ref{fig:availability_region}. This point is achievable by time sharing between strategies that achieve availabilities corresponding either to points $A$ and $B$ or $C$ and $D$, which are all achievable as argued in Remark~\ref{rem:achievability}. This clearly shows that there are numerous different ways with which $\rho \in \nfrakR$ is achievable. To show that the point $\rho \notin \nfrakR$ is not achievable, consider point $F = (\rho^F_1, \rho^F_2)$ in Fig.~\ref{fig:availability_region}. Note that given $\rho^F_1$, the maximum availability possible for second tier is constrained by the corresponding value $\rho_2^*(\rho_1^F)$ on the solid curve. Since $\rho^F_2 > \rho_2^*(\rho_1^F)$, it contradicts the fact that $\rho_2^*(\rho_1^F)$ is the maximum possible availability for second tier given $\rho_1^F$. Hence point $F$ is not achievable.
\end{proof}

\begin{figure}
\centering
\includegraphics[width=\columnwidth]{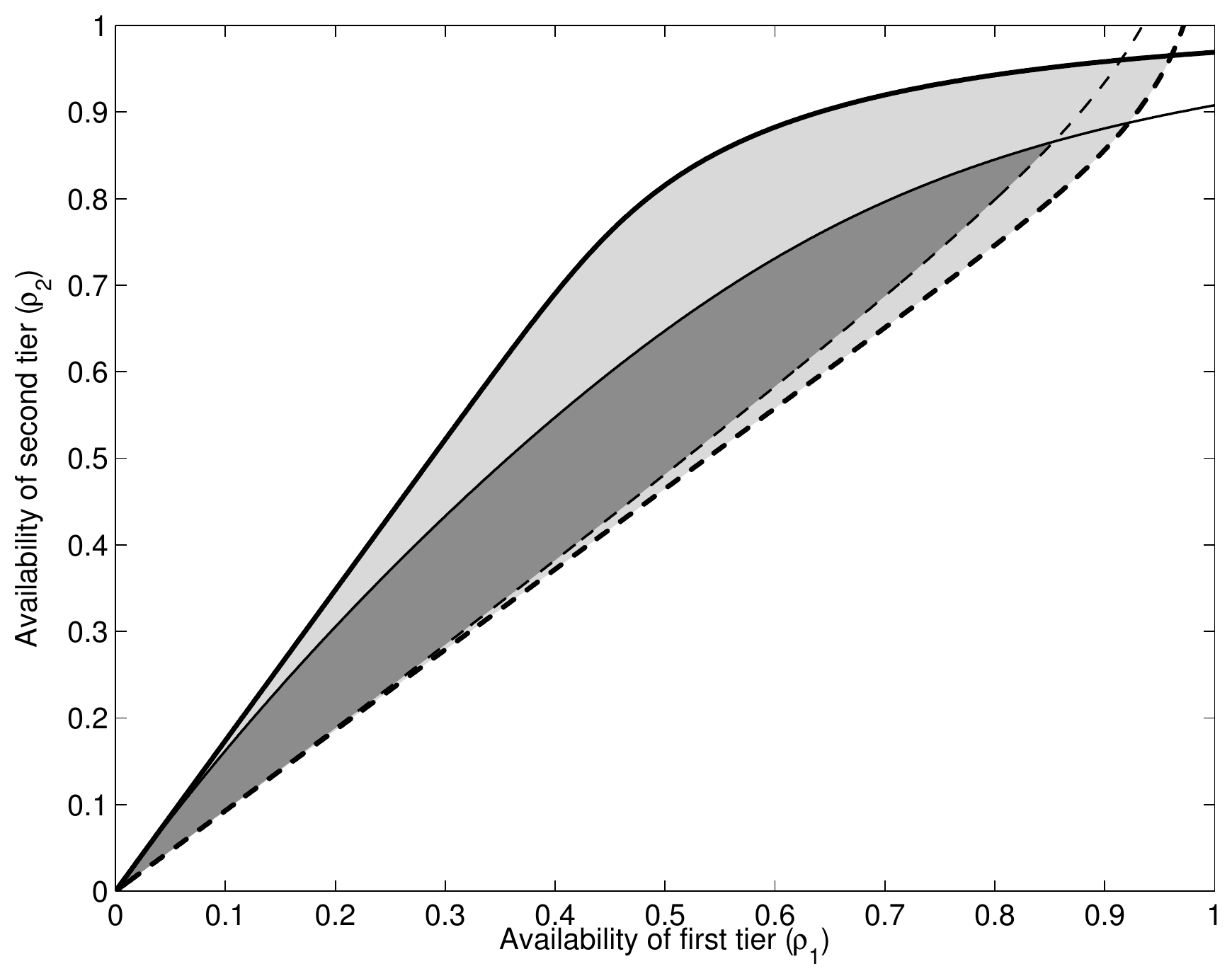}
\caption{Availability region for a two-tier HetNet is denoted by lightly shaded region. The availability region when one of the tiers is constrained to use $\ncalS_k(N_k)$ is denoted by the dark shade. Setup: $\alpha=4, K=2, N_1 = 20, N_2 = 15, \gamma=1.1, \mu_1=15, \mu_2=5, \lambda_2 = 10\lambda_1$, $m_1 = m_2$, $\sigma_1 = \sigma_2$.}
\label{fig:availability_constrained}
\end{figure}

\begin{remark}[Effect of constraining the set of strategies on $\nfrakR$] \label{rem:availability_constrained}
Recall that $\rho_k^*$ given by~\eqref{eq:rhok_star} and used in defining the availability region $\nfrakR$ corresponds to fixed point solution for strategy $\ncalS(1)$. In principle, it is possible to restrict one of tiers to follow a particular strategy by defining $\rho_k^*$ as the fixed point solution for that strategy. For instance, we could define $\rho_k^*$ as a solution to the fixed point equation corresponding to strategy $\ncalS_k(N_k)$. Clearly, all the points $\rho \in \nfrakR$ will not be achievable in this setup. For a two tier setup, we plot the availability region for this case in Fig.~\ref{fig:availability_constrained}, along with the availability region $\nfrakR$ defined by Theorem~\ref{thm:availability_region}. Note that as expected the set of points achievable under this constrained setup is strictly contained in the availability region defined by Theorem~\ref{thm:availability_region}. 
\end{remark}

We conclude this subsection with two remarks about the ``optimality'' of the availability region.

\begin{remark}[Higher availability is not always better]
It is not always optimal in terms of certain performance metrics to operate the network in the regime corresponding to the maximum availabilities. We will validate this in Section~\ref{sec:numresults} in terms of the downlink rate. Interestingly, a similar idea, although applicable at a much smaller time scale, of intentionally making a macrocell ``unavailable'' on certain sub-frames can be used to improve downlink data rate by offloading more users to the small cells. This concept is called almost blank sub-frames (ABS) and was introduced as a part of enhanced inter-cell interference coordination (eICIC) in 3GPP LTE release 10~\cite{3GPPM2012a}. While this is an interesting analogy, the two concepts are not exactly the same because in addition to the differences in the time scales, ABS additionally assumes coordination across BSs.
\end{remark}

\begin{remark}[Notion of optimality]
The performance of a HetNet with energy harvesting is fundamentally the same as the one with a reliable energy source if for the given performance metric, the optimal availabilities $\hat{\rho}$ lie in the availability region, i.e., $\hat{\rho} \in \nfrakR$. For example, if $\hat{\rho}$ corresponds to point $E$ in Fig.~\ref{fig:availability_region}, the HetNet despite having unreliable energy source will achieve ``optimal'' performance. On the other hand, if $\hat{\rho}$ is, say, point $F$ in Fig.~\ref{fig:availability_region}, there will be some performance loss due to unreliability in energy availability.
\end{remark}

We now study the coverage probability and downlink rate in the following subsection, which will be useful in the next section to demonstrate the above ideas about optimality.

\subsection{Coverage Probability and Downlink Rate} \label{sec:performance}
We now study the effect of BS availabilities $\{\rho_k\}$ on the downlink performance at small time scale. As described in Section~\ref{sec:sysmod}, the availabilities change on a much longer time scale and hence the operational states of the BSs can be considered static over small time scale. Therefore, for this discussion it is enough to consider the set of available BSs $\PPPa$. For downlink analysis, we focus on a typical user assumed to be located at the origin, which is made possible by Skivnyak's theorem~\cite{StoKenB1995}. Assuming full-buffer model for inter-cell interference~\cite{DhiGanJ2012}, i.e., all the interfering BSs in $\PPPa$ are always active, the $\sir$ at a typical user when it connects to a BS located at $x \in \PPPak$ is
\begin{align}
\sir(x) = \frac{P_k h_{kx}^{(0)} \ncalX_{kx}^{(0)} \|x\|^{-\alpha}}{\sum_{j \in \ncalK} \sum_{z \in \PPPaj\setminus \{x\}} P_j h_{jz}^{(0)} \ncalX_{jz}^{(0)} \|z\|^{-\alpha}}.
\end{align}
Using tools developed in~\cite{DhiAndJ2013}, Theorem 1 of~\cite{JoSanJ2012} can be easily extended to derive the coverage probability under the general cell selection model of this paper, which additionally incorporates the effect of shadowing. Since the extension is straightforward, the proof is skipped.
\begin{theorem}[Coverage] \label{cor:Pc}
The coverage probability is
\begin{align}
\pc = \nbbP(\sir(x^{*(0)}) > \T) = \frac{1}{1+\ncalF(\T, \alpha)},
\label{eq:pc_int_limited}
\end{align}
where
\begin{align}
\ncalF(\T, \alpha) &= \left( \frac{2 \T}{\alpha - 2} \right) {}_2F_1 \left[1,1-\frac{2}{\alpha}, 2 - \frac{2}{\alpha}, -\T \right],
\end{align}
and ${}_2F_1[a,b,c,z] = \frac{\Gamma(c)}{\Gamma(b) \Gamma(c-b)} \int_{0}^1 \frac{t^{b-1} (1-t)^{c-b-1}}{(1-tz)^a} {\rm d} t$ denotes Gauss hypergeometric function.
\end{theorem}
Clearly, the coverage probability for interference-limited HetNets is independent of the densities of the available BSs, and hence of the availabilities $\{\rho_k\}$. This validates Remark~\ref{rem:Pc}. However, it is not necessarily so in the case of downlink rate distribution, which we discuss next. Assuming equal resource allocation across all the users served by a BS, the complimentary cumulative distribution function (CCDF) of rate $\ncalR$ (in bps/Hz) achieved by a typical user, termed {\em rate coverage} $\rc$, is calculated in~\cite{DhiAndJ2013} for the same cell selection model as this paper. Assuming the typical user connects to a $k^{th}$ tier BS, $\ncalR$ can be expressed as $\ncalR = \frac{1}{\Psi_k} \log(1 + \sir(x^{*(0)})$, where $\Psi_k$ is the number of users served by the $k^{th}$ tier BS to which the typical user is connected. The approach of~\cite{DhiAndJ2013} includes approximating the distribution of $\Psi_k$ and assuming it to be independent of $\sir(x^{*(0)})$ to derive an accurate approximation of $\ncalR$. With two minor modifications, i.e., the density of $k^{th}$ tier active BSs is $\rho_k \lambda_k$, and the effective density of active users is $\pc \lambda_u$, the result of~\cite{DhiAndJ2013} can be easily extended to the current setup and is given in the following theorem. For proof and other related details, please refer to Section III of~\cite{DhiAndJ2013}. This result will be useful in demonstrating the fact that the optimal downlink performance may not always correspond to the regime of maximum availabilities.  
\begin{theorem}[Rate CCDF] \label{thm:Rc}
The CCDF of downlink rate $\ncalR$ (in bps/Hz) or rate coverage $\rc$ is
\begin{align}
&\nbbP(\ncalR > \ncalT) = \sum_{n\geq 0} \frac{1}{1 +  \ncalF\left(\T_{n+1}, \alpha \right)} \sum_{k=1}^K  
 \frac{\rho_k  \lambda_k \nbbE \left[\ncalX_k^{\frac{2}{\alpha}} \right] P_k^{\frac{2}{\alpha}}}
{\sum\limits_{j\in\ncalK} \rho_j \lambda_j \nbbE \left[\ncalX_j^{\frac{2}{\alpha}} \right] P_j^{\frac{2}{\alpha}}   }  \nonumber   \\
&\times \frac{3.5^{3.5}}{n!} \frac{\Gamma(n+4.5)}{\Gamma(3.5)} \left(\frac{\pc\lambda_u \ncalP_k}{\rho_k \lambda_k} \right)^n \left(3.5 +  \frac{\pc\lambda_u \ncalP_k}{\rho_k \lambda_k}  \right)^{-(n+4.5)} \nonumber
\end{align}
where $\T_{n+1} = 2^{\ncalT (n+1)} -1$ and 
\begin{align}
\ncalP_k = \frac{ \rho_k \lambda_k \nbbE \left[\ncalX_k^{\frac{2}{\alpha}} \right]  P_k^{\frac{2}{\alpha}}}{\sum_{j \in \ncalK}  \rho_j  \lambda_j \nbbE \left[\ncalX_j^{\frac{2}{\alpha}} \right]  P_j^{\frac{2}{\alpha}}}.
\end{align}
\end{theorem}
\begin{remark} [Invariance to shadowing distribution]   \label{rem:shadowinvar_Rate}
From Theorem~\ref{thm:Rc}, we note that the rate coverage is invariant to the shadowing distribution when $\nbbE \left[\ncalX_j^{\frac{2}{\alpha}} \right] = \nbbE \left[\ncalX_k^{\frac{2}{\alpha}} \right]$, for all $j, k \in \ncalK$. This is similar to the observations made in Remark~\ref{rem:shadowinvar_EnUtil}. 
\end{remark}

\section{Numerical Results and Discussion} \label{sec:numresults}
Since most of the analytical results discussed in this paper are self-explanatory, we will focus only on the most important trends and insights in this section. For conciseness, we assume lognormal shadowing for each tier with the same mean $m$ dB and standard deviation $\sigma$ dB. Recall that both the energy utilization and the rate distribution results are invariant to shadowing under this assumption, as discussed in Remarks~\ref{rem:shadowinvar_EnUtil} and~\ref{rem:shadowinvar_Rate}. We begin by discussing the effect of battery capacity on the availability region.

\begin{figure}
\centering
\includegraphics[width=\columnwidth]{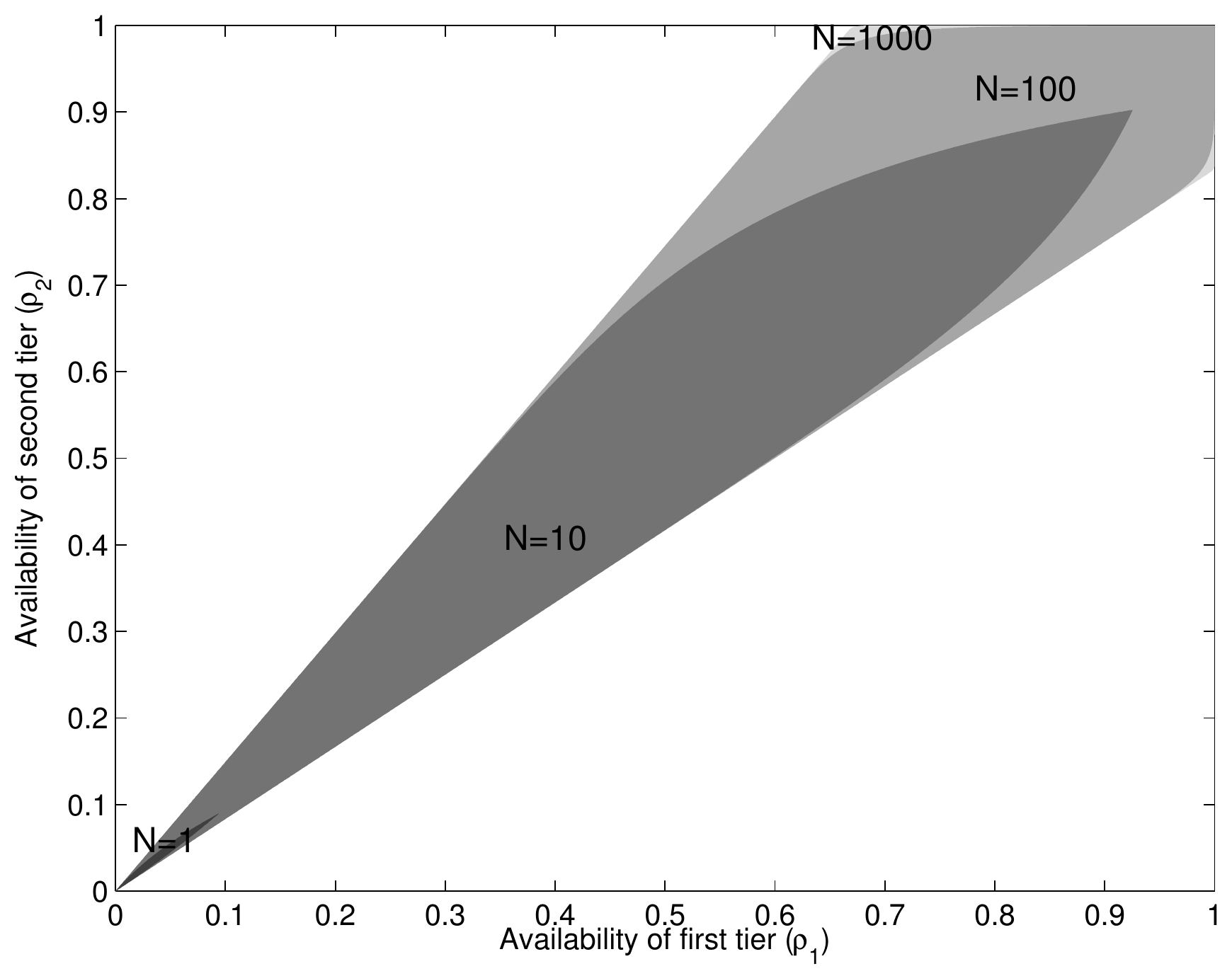}
\includegraphics[width=\columnwidth]{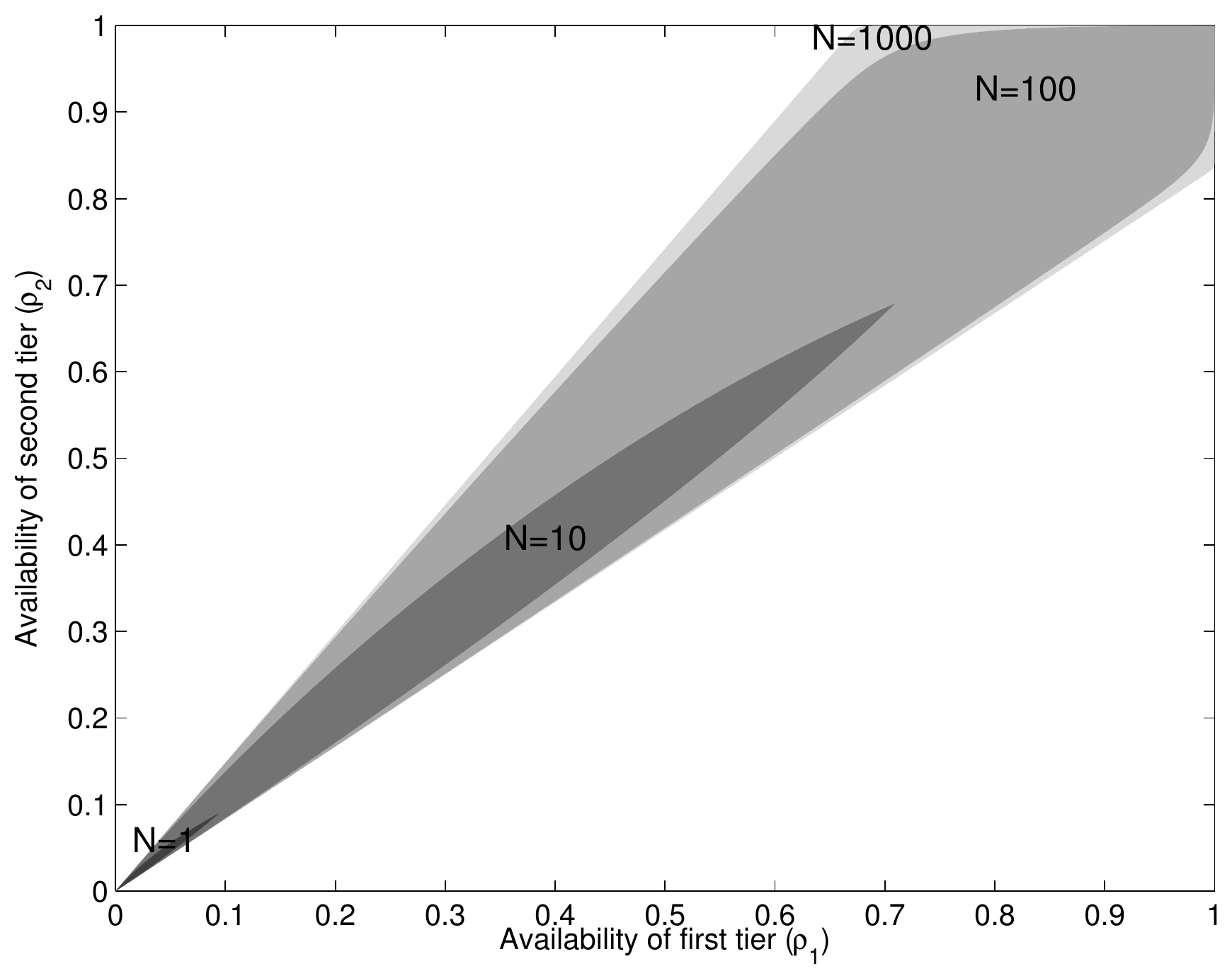}
\caption{{\em (first)} Availabilities region for various values of energy storage capacity $N$, where $N_1 = N_2 = N$. {\em (second)} One of the tiers constrained to use strategy $\ncalS_k(N)$. Setup: $\alpha=4, K=2, \gamma = 1.1, P = [1, 0.1], \mu_1=10, \mu_2=3, \lambda_2 = 10\lambda_1$.}
\label{fig:avail_N}
\end{figure}

\subsection{Effect of Battery Capacity on Availability Region}
We consider a two tier HetNet and plot its availability region for various values of the capacity of the energy storage module, i.e., battery capacity, in the first subplot of Fig~\ref{fig:avail_N}. For ease of exposition, we assume that the storage capacities of the BSs of the two tiers are the same. As expected, the availability region $\nfrakR$ increases with the increase in battery capacity. Interestingly, it is however not possible to achieve all the points $\rho$ in the square $[0,1]\times [0,1]$ even by increasing the battery capacity infinitely. The maximum availability region is a function of over-provisioning factor $\gamma$, which is set to $1.1$ for this result. Additionally, we note that the maximum availabilities for both the tiers approach unity even at modest battery levels. We repeat the same experiment for the case when one of the tiers is constrained to use the strategy $\ncalS_k(N)$ and present the results in the second subplot of Fig.~\ref{fig:avail_N}. Recall that this case was discussed in Remark~\ref{rem:availability_constrained}. We observe that for the same battery capacity $N$, the achievable region is smaller in this case compared to Fig.~\ref{fig:avail_N}, which is consistent with the observations made in Section~\ref{sec:availability}. The difference is especially prominent for smaller values of battery capacity $N$.

\begin{figure}
\centering
\includegraphics[width=\columnwidth]{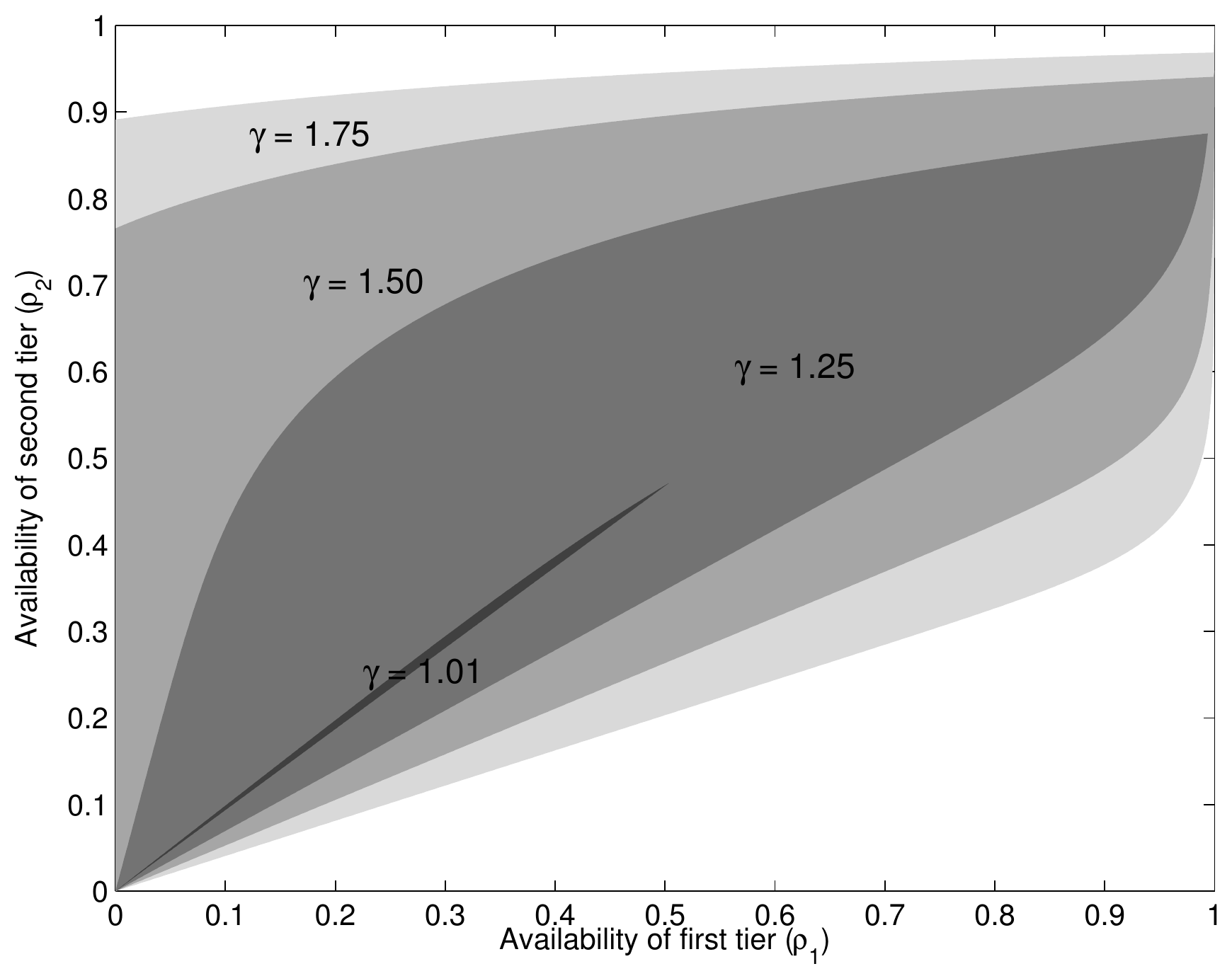}
\includegraphics[width=\columnwidth]{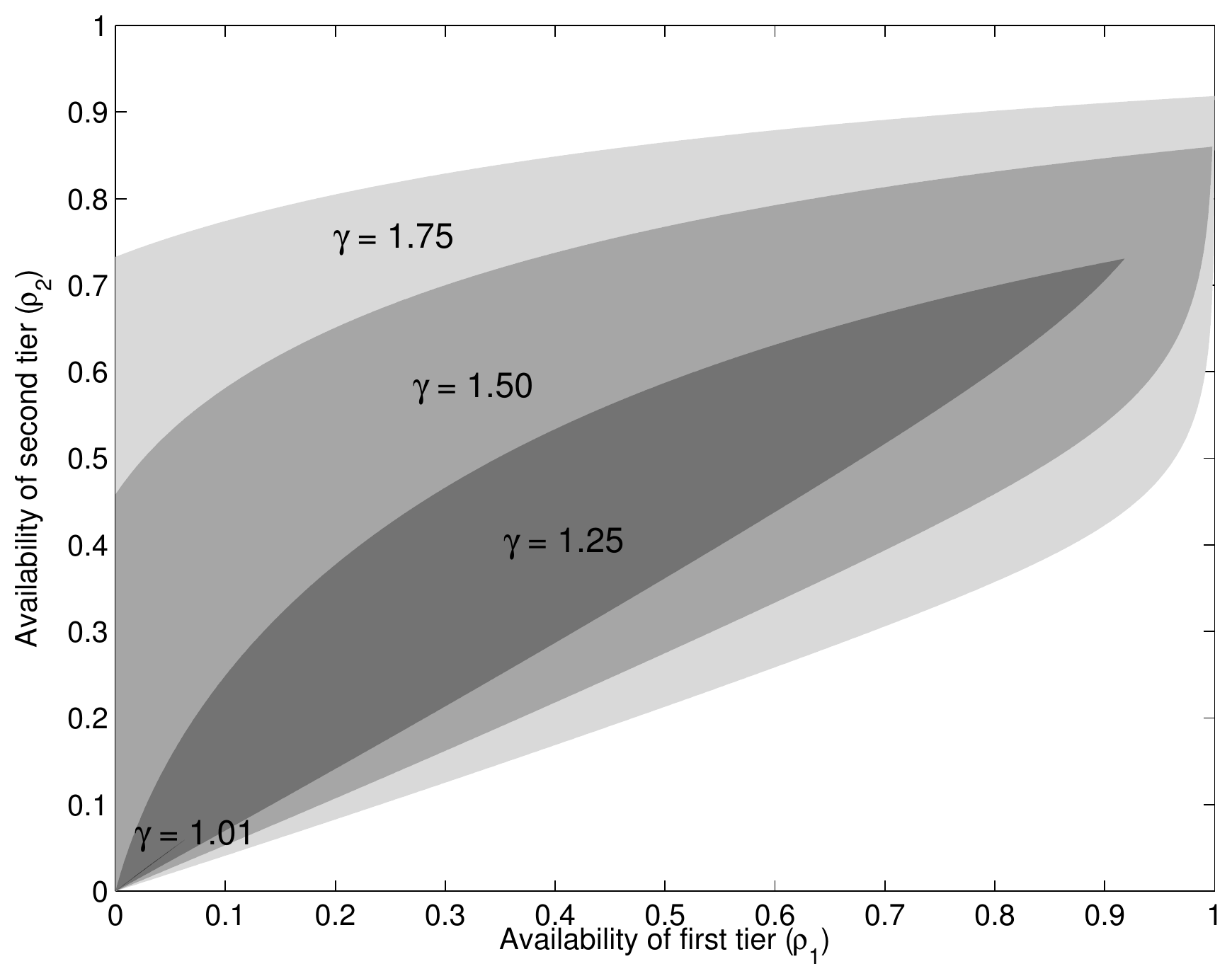}
\caption{{\em (first)} Availabilities region for various values of $\gamma$. {\em (second)} One of the tiers is constrained to use strategy $\ncalS_k(N)$. Setup: $\alpha=4, K=2, N_1 = 20, N_2 = 5, \mu_1=10, \mu_2=3, \lambda_2 = 10\lambda_1$.}
\label{fig:avail_Gamma}
\end{figure}

\subsection{Effect of Over-Provisioning Factor on Availability Region}
We now study the effect of the over-provisioning factor on the availability region in the first subplot of Fig.~\ref{fig:avail_Gamma}. Recall that the over-provisioning factor $\gamma$ is the ratio of the net energy harvested per unit area per unit time and the net energy utilized per unit area per unit time. The first and foremost observation is that unlike increasing battery capacity, the availability region expands by increasing $\gamma$ and will cover the complete square $[0,1]\times[0,1]$ for sufficiently large $\gamma$. Also note that the beyond a certain value of $\gamma$, the availability of a tier may be non-zero even if the availabilities of the other tiers are zero. This is the case when that tier harvests enough energy on its own to serve all the load offered to the network, i.e., $\lambda_k \mu_k > \pc \lambda_u$. As in the previous subsection, we now repeat this experiment under the constraint that one of the tiers follows strategy $\ncalS_k(N_k)$ and present the results in the second subplot of Fig.~\ref{fig:avail_Gamma}. As expected, the availability region is considerably smaller in this case.

\subsection{Rate coverage}
Using rate coverage, given by Theorem~\ref{thm:Rc}, we demonstrate that it may not always be optimal to operate the network in the regime corresponding to maximum availabilities. We plot rate coverage as a function of $(\rho_1, \rho_2)$ for two different setups in Fig.~\ref{fig:RcMesh1}. In both the cases, we note that it is strictly suboptimal to operate at the point $(\rho_1, \rho_2) = (1,1)$. Furthermore, as the second tier density is increased, it is optimal to keep first tier BSs OFF more often. As expected, the rate coverage also increases with the increase of second tier density. This example additionally motivates the need for the exact characterization of $\hat{\rho}$ for various metrics of interest, which forms a concrete line of future work. Once the optimal $\hat{\rho}$ for a given metric is known, the system designers can, in principle, design the energy harvesting modules such that $\hat{\rho} \in \nfrakR$. In such a case, the HetNet with energy harvesting will have fundamentally ``optimal'' performance, i.e., the same performance as the HetNet with reliable energy sources. 

\begin{figure}
\centering
\includegraphics[width=\columnwidth]{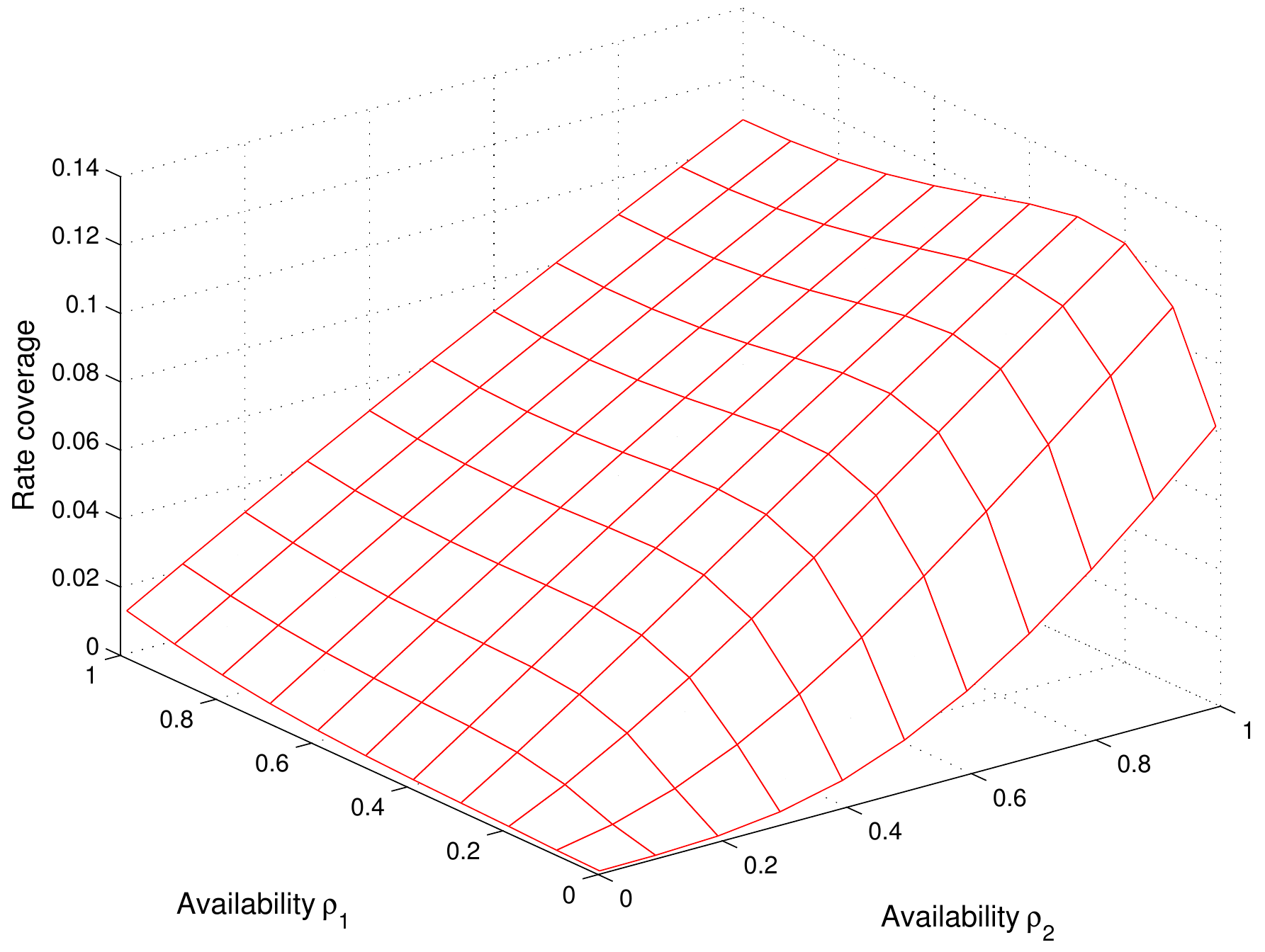}
\includegraphics[width=\columnwidth]{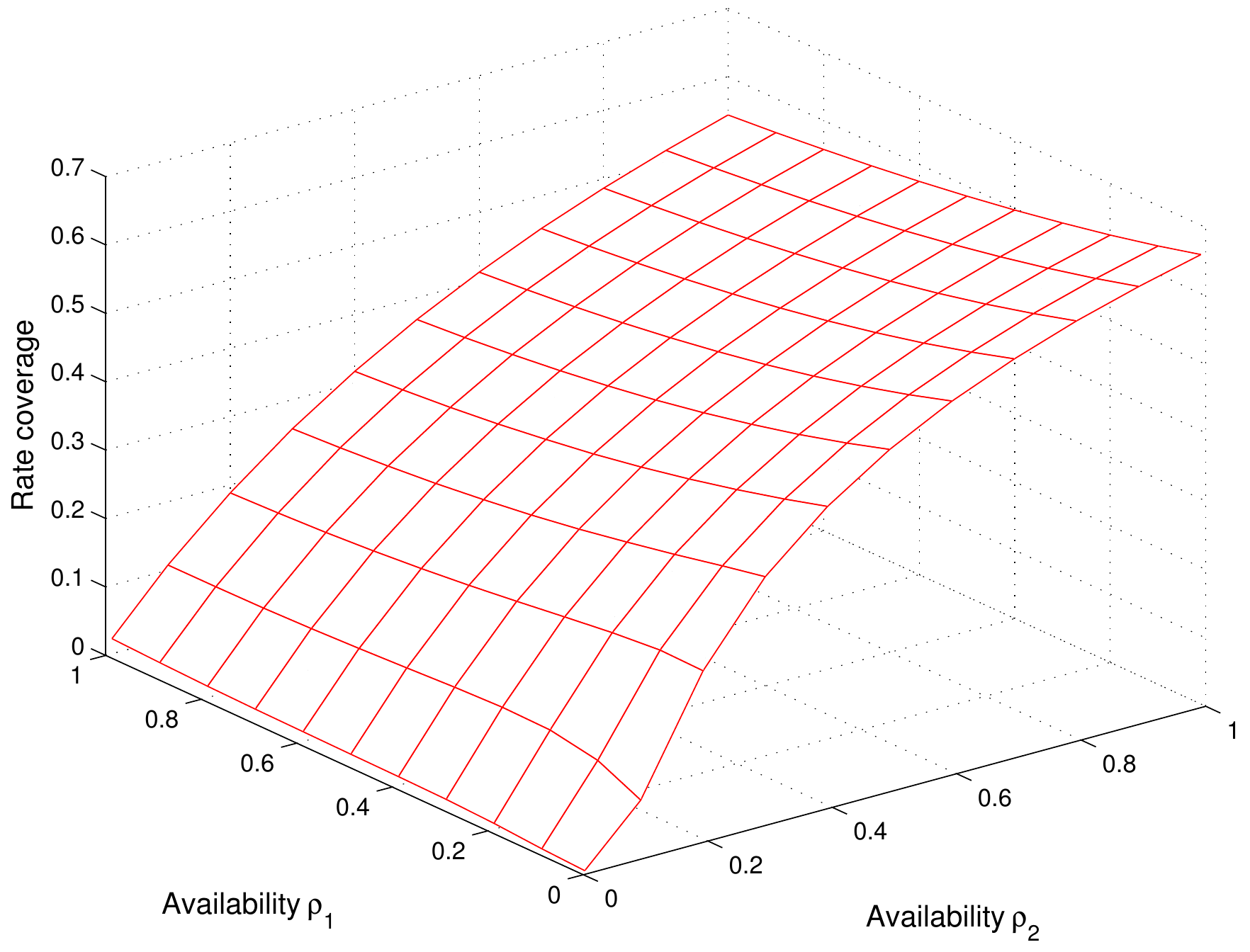}
\caption{Rate coverage as a function of $\rho_1$ and $\rho_2$. Setup: $\alpha=4, K=2, P = [1, 0.01], \ncalT = 0.1, \lambda_u = 100\lambda_1$. {\em (first)} $\lambda_2 = 2\lambda_1$, {\em (second)} $\lambda_2 = 20\lambda_1$. Note that $\rho=[1\ 1]$ is not optimal in both the cases.}
\label{fig:RcMesh1}
\end{figure}

\section{Conclusions}
In this paper, we have developed a comprehensive framework to study HetNets, where each BS is powered solely by its energy harvesting module. Developing novel tools with foundations in random walk theory, fixed point analysis and stochastic geometry, we quantified the uncertainty in BS availability due to the finite battery capacity and inherent randomness in energy harvesting. We further characterized the availability region for a set of general uncoordinated BS operational strategies. This provides a fundamental characterization of the regimes under which the HetNets with energy harvesting modules are fundamentally optimal, i.e., have the same performance as the ones with reliable energy sources. 

This work has many extensions. From modeling perspective, it is important to incorporate more accurate energy expenditure models taking into account energy spent on backhaul and control signaling, and model the energy required to transmit a packet to a user as a function of the $\sir$. It is also important to extend the developed framework to study coordinated strategies. From optimality perspective, it is important to characterize the optimal availabilities as a function of key system parameters for various metrics, such as the downlink rate, so that the energy harvesting modules can be accordingly designed. From physical layer perspective, more sophisticated transmission techniques, such as MIMO, should be taken into account, e.g., using tools developed in~\cite{DhiKouJ2013,LiZhaJ2013}. From cellular perspective, it is desirable to consider the effect of unreliable energy sources on uplink, e.g., using tools from~\cite{NovDhiJ2012}.

\appendix

\subsection{Proof of Lemma~\ref{lem:avgservicearea}} \label{Appendix:avgservicearea}
Denote by $\nbbP_{\PPPak}^{x_k}(\cdot)$ and $\nbbE_{\PPPak}^{x_k}[\cdot]$, the conditional (Palm) probability and conditional expectation, conditioned on $x_k \in \PPPak$. Please refer to~\cite{BacBlaB2009,BacBreB2003,StoKenB1995} for details on Palm calculus. Before we derive the average area, note that for given realizations of the BS locations and the channel gains, the area of the service region of the $k^{th}$ tier BS is
\begin{align}
|\ncalA_k(x_k)| = \int\limits_{\nbbR^2} \prod_{j \in \ncalK} \prod_{x \in \PPPaj } \nb1 \left(\frac{P_k \ncalX_k^{(z)}} { \|x_k - z\|^{\alpha}} \geq \frac{P_j \ncalX_j^{(z)}} { \|x - z\|^{\alpha} } \right) {\rm d} z.
\end{align}
The average service area can now be expressed as
\begin{align}
\nbbE[|\ncalA_k(x_k)|] &\stackrel{(a)}{=} \nbbE \nbbE_{\PPPak}^{x_k}[|\ncalA_k(x_k)|]\\
&\stackrel{(b)}{=} \nbbE \nbbE_{\PPPak}^{0}[|\ncalA_k(0)|] \stackrel{(c)}{=} \nbbE \nbbE_{\PPPak}[|\ncalA_k(0)|], \label{eq:averagearea_int1}
\end{align}
where $(a)$ follows by distributing the expectation over the point process $\PPPak$ and the rest of the randomness, $(b)$ follows from the stationarity of the homogeneous PPP, and $(c)$ follows from Slivnyak's theorem~\cite{StoKenB1995}. Substituting the expression for $|\ncalA_k(0)|$ in \eqref{eq:averagearea_int1} and distributing the expectation across various random quantities, we can express the average area as $\nbbE[|\ncalA_k(x_k)|] =$
\begin{align}
\nbbE_{\ncalX_k} \int\limits_{\nbbR^2} \prod_{j \in \ncalK} \nbbE_{\PPPaj} \prod_{x \in \PPPaj } \nbbE_{\ncalX_j} \nb1 \left(\frac{P_k \ncalX_k} { \|z\|^{\alpha}} \geq \frac{P_j \ncalX_j} { \|x - z\|^{\alpha} } \right) {\rm d} z,
\end{align}
where the expectations over point processes $\PPPaj$ and shadowing gains $\ncalX_j$ can be moved inside respective product terms due to independence, and superscript on $\ncalX_k^{(z)}$ and $\ncalX_j^{(z)}$ are removed for notational simplicity. The expectation over point process $\PPPaj$ can be evaluated using the probability generating functional (PGFL)~\cite{StoKenB1995}, which simplifies the average area expression to
\begin{align}
\nbbE_{\ncalX_k} \int\limits_{\nbbR^2} \prod_{j \in \ncalK} e^{ -\rho_j \lambda_j \nbbE_{\ncalX_j} \int_{\nbbR^2} \nb1 \left(\frac{P_k \ncalX_k} { \|z\|^{\alpha}} < \frac{P_j \ncalX_j} { \|x - z\|^{\alpha} } \right) {\rm d}x } {\rm d} z.
\end{align}
Solving the integral in the exponential, we get
\begin{align}
\nbbE_{\ncalX_k} \int\limits_{\nbbR^2} \prod_{j \in \ncalK} \exp\left( -\rho_j \lambda_j \pi \|z\|^2 \frac{P_j \nbbE \left[\ncalX_j^{\frac{2}{\alpha}} \right]}{P_k \ncalX_k^{\frac{2}{\alpha}}}  \right) {\rm d} z,
\end{align}
which can be equivalently expressed as
\begin{align}
\nbbE_{\ncalX_k} \int\limits_{\nbbR^2}  \exp\left( -\rho_j \lambda_j \pi \|z\|^2 \frac{\sum_{j \in \ncalK} P_j \nbbE \left[\ncalX_j^{\frac{2}{\alpha}} \right]}{P_k \ncalX_k^{\frac{2}{\alpha}}}  \right) {\rm d} z,
\end{align}
from which the result follows by solving the integral and taking expectation with respect to $\ncalX_k$. 
\hfill \QED

\subsection{Proof of Lemma~\ref{lem:properties}} \label{Appendix:properties}
Since both the properties (monotonicity and concavity) are element-wise properties, it is enough to consider the given function as a function of single variable $x \in \nbbR$. After dropping the subscript $k$ and with slight overloading of notation, we denote this function as $g: \nbbR \rightarrow \nbbR$, which is
\begin{align}
g(x) = 1 - \left(\frac{1 - b - ax}{1-(b+ax)^N} \right),
\end{align}
where $b \in \nbbR_+$ is a constant when we study element-wise properties. We now do the following substitution $x + \frac{b}{a} \rightarrow x$, which just shifts the function along $x$-axis and hence neither impacts the monotonicity nor concavity of $g(x)$. The simplified expression is
\begin{align}
g(x) = 1 - \left(\frac{1 - ax}{1-(ax)^N} \right).
\end{align}
Note that although both the numerator and denominator of the second term in the above expression go to $0$ as $x \rightarrow \frac{1}{a}$, it is easy to show that the function is continuous at this point and the limit is
\begin{align}
\lim_{x \rightarrow \frac{1}{a}} g(x) = \frac{N-1}{N}.
\end{align}
To prove that the function is monotonically increasing, it is enough to show that the partial derivative with respect to $x$ is positive. The partial derivative can be expressed as
\begin{align}
g'(x) = \frac{a}{(1-(ax)^N)^2}\left((N-1)(ax)^N +1 - N(ax)^{N-1} \right),
\end{align}
It is easy to show that the term inside the bracket is positive except at $x = \frac{1}{a}$, where it has a minima and takes value $0$. Further, using L'H\^{o}pital's rule it is straightforward to show
\begin{align}
\lim_{x \rightarrow \frac{1}{a}} g'(x) = a \frac{N-1}{2N} > 0,
\end{align} 
which completes the proof for the monotonicity property. To show that the function is concave, we need to show that the double derivative with respect to $x$ is negative, which is
\begin{align}
&g''(x) = -a^2 N (ax)^{N-2}\frac{1-ax}{(1-(ax)^N)^3}\ \times \nonumber \\
&\left(\frac{(N-1)(1-(ax)^{N+1})}{1-ax} - \frac{ax(N+1)(1-(ax)^{N-1})}{1-ax} \right),
\end{align}
where the term inside the bracket is positive except at $x=\frac{1}{a}$, where it has a minima and takes value $0$. As in the case of the first derivative, it is easy to show using L'H\^{o}pital's rule that the limit at this point is
\begin{align}
\lim_{x \rightarrow \frac{1}{a}} g''(x) = -a^2 \frac{N^2-1}{6N} < 0,
\end{align}
which shows that the function is strictly concave for all $x \in \nbbR$. This completes the proof. \hfill \QED

\subsection{Proof of Lemma~\ref{lem:equivalence}} \label{Appendix:equivalence}
For the proof of \eqref{eq:sol_existence} $\Rightarrow$ \eqref{eq:sol_condition}, take the denominator of \eqref{eq:sol_existence} to the right hand side of inequality and multiply both sides by $\lambda_k$ to get
\begin{align}
\lambda_k \mu_k \sum\limits_{j =1}^K  \rho_j \lambda_j \nbbE \left[\ncalX_j^{\frac{2}{\alpha}} \right] P_j^{\frac{2}{\alpha}} > \rho_k \lambda_k \pc \lambda_u \nbbE \left[\ncalX_k^{\frac{2}{\alpha}} \right] P_k^{\frac{2}{\alpha}}, \forall k \in \ncalK. \nonumber
\end{align}
Now add all the $K$ inequalities, i.e., sum both sides from $k=1$ to $K$, which leads to \eqref{eq:sol_condition} and hence completes half of the proof. For the proof of \eqref{eq:sol_existence} $\Leftarrow$ \eqref{eq:sol_condition}, multiply both sides of~\eqref{eq:sol_condition} by $\sum_{j=1}^K  \rho_j \lambda_j  \nbbE \left[\ncalX_j^{\frac{2}{\alpha}} \right] P_j^{\frac{2}{\alpha}}$ to get

\begin{align}
\sum_{k=1}^{K} \lambda_k \mu_k  \sum_{j=1}^K  \rho_j \lambda_j  \nbbE \left[\ncalX_j^{\frac{2}{\alpha}} \right] P_j^{\frac{2}{\alpha}} > \sum_{k=1}^K \pc \lambda_u \rho_k \lambda_k  \nbbE \left[\ncalX_k^{\frac{2}{\alpha}} \right] P_j^{\frac{2}{\alpha}}.
\end{align}
Rearranging the terms we get

\begin{align}
\sum_{k=1}^K \lambda_k \left(\frac{\mu_k  \sum_{j=1}^K  \rho_j \lambda_j  \nbbE \left[\ncalX_j^{\frac{2}{\alpha}} \right] P_j^{\frac{2}{\alpha}}  }   {\pc \lambda_u \rho_k  \nbbE \left[\ncalX_k^{\frac{2}{\alpha}} \right] P_j^{\frac{2}{\alpha}} } - 1 \right) > 0.
\end{align}
Since $\lambda_k$ is arbitrary, for the above condition to always hold, we need the term inside the bracket to be positive for all $k \in \ncalK$. This set of conditions is the same as \eqref{eq:sol_existence} and hence completes the proof.
\hfill \QED

\bibliographystyle{IEEEtran}
\bibliography{EnergyHarvesting_Journal_v1.6}

\begin{thebibliography}{10}
\providecommand{\url}[1]{#1}
\csname url@samestyle\endcsname
\providecommand{\newblock}{\relax}
\providecommand{\bibinfo}[2]{#2}
\providecommand{\BIBentrySTDinterwordspacing}{\spaceskip=0pt\relax}
\providecommand{\BIBentryALTinterwordstretchfactor}{4}
\providecommand{\BIBentryALTinterwordspacing}{\spaceskip=\fontdimen2\font plus
\BIBentryALTinterwordstretchfactor\fontdimen3\font minus
  \fontdimen4\font\relax}
\providecommand{\BIBforeignlanguage}[2]{{%
\expandafter\ifx\csname l@#1\endcsname\relax
\typeout{** WARNING: IEEEtran.bst: No hyphenation pattern has been}%
\typeout{** loaded for the language `#1'. Using the pattern for}%
\typeout{** the default language instead.}%
\else
\language=\csname l@#1\endcsname
\fi
#2}}
\providecommand{\BIBdecl}{\relax}
\BIBdecl

\bibitem{DhiLiC2013}
H.~S. Dhillon, Y.~Li, P.~Nuggehalli, Z.~Pi, and J.~G. Andrews, ``Fundamentals
  of base station availability in cellular networks with energy harvesting,''
  in {\em IEEE Globecom}, Atlanta, GA, Dec. 2013.

\bibitem{AndJ2013}
J.~G. Andrews, ``Seven ways that {HetNets} are a cellular paradigm shift,''
  \emph{IEEE Communications Magazine}, vol.~51, no.~3, pp. 136 -- 144, Mar.
  2013.

\bibitem{DhiGanJ2013}
H.~S. Dhillon, R.~K. Ganti, and J.~G. Andrews, ``Load-aware modeling and
  analysis of heterogeneous cellular networks,'' \emph{IEEE Trans. on Wireless
  Communications}, vol.~12, no.~4, pp. 1666 -- 1677, Apr. 2013.

\bibitem{WorCenM2006}
Worldwatch Institute and Center for American Progress, ``American Energy: The
  Renewable Path to Energy Security'', Sep. 2006.

\bibitem{HurKimC2011}
S.~Hur, T.~Kim, D.~J. Love, J.~V. Krogmeier, T.~A. Thomas, and A.~Ghosh,
  ``Multilevel millimeter wave beamforming for wireless backhaul,'' in
  \emph{Proc., IEEE Globecom Workshops}, Houston, TX, 2011.

\bibitem{HoZhaJ2012}
C.~K. Ho and R.~Zhang, ``Optimal energy allocation for wireless communications
  with energy harvesting constraints,'' \emph{IEEE Trans. on Signal
  Processing}, vol.~60, no.~9, pp. 4808 -- 4818, Sep. 2012.

\bibitem{OzeTutJ2011}
O.~Ozel, K.~Tutuncuoglu, J.~Yang, S.~Ulukus, and A.~Yener, ``Transmission with
  energy harvesting nodes in fading wireless channels: Optimal policies,''
  \emph{IEEE Journal on Sel. Areas in Communications}, vol.~29, no.~8, pp. 1732
  -- 1743, Sep. 2011.

\bibitem{TutYenJ2012}
K.~Tutuncuoglu and A.~Yener, ``Optimum transmission policies for battery
  limited energy harvesting nodes,'' \emph{IEEE Trans. on Wireless
  Communications}, vol.~11, no.~3, pp. 1180 -- 1189, Mar. 2012.

\bibitem{ShaMukJ2010}
V.~Sharma, U.~Mukherji, V.~Joseph, and S.~Gupta, ``Optimal energy management
  policies for energy harvesting sensor nodes,'' \emph{IEEE Trans. on Wireless
  Communications}, vol.~9, no.~4, pp. 1326 -- 1336, Apr. 2010.

\bibitem{YanUluJ2012}
J.~Yang and S.~Ulukus, ``Optimal packet scheduling in an energy harvesting
  communication systems,'' \emph{IEEE Trans. on Communications}, vol.~60,
  no.~1, pp. 220 -- 230, Jan. 2012.

\bibitem{YanOzeJ2012}
J.~Yang, O.~Ozel, and S.~Ulukus, ``Broadcasting with an energy harvesting
  rechargeable transmitter,'' \emph{IEEE Trans. on Wireless Communications},
  vol.~11, no.~2, pp. 571 -- 583, Feb. 2012.

\bibitem{AntUysJ2011}
M.~Antepli, E.~Uysal-Biyikoglu, and H.~Erkal, ``Optimal packet scheduling on an
  energy harvesting broadcast link,'' \emph{IEEE Journal on Sel. Areas in
  Communications}, vol.~29, no.~8, pp. 1721 -- 1731, Aug. 2011.

\bibitem{GatGeoJ2010}
M.~Gatzianas, L.~Georgiadis, and L.~Tassiulas, ``Control of wireless networks
  with rechargeable batteries,'' \emph{IEEE Trans. on Wireless Communications},
  vol.~9, no.~2, pp. 581 -- 593, Feb. 2010.

\bibitem{HuaJ2013}
K.~Huang, ``Spatial throughput of mobile ad hoc networks with energy
  harvesting,'' submitted to {\em IEEE Trans. on Info. Theory}. Availabile
  online: arxiv.org/abs/1111.5799.

\bibitem{HuaC2011}
------, ``Throughput of wireless networks powered by energy harvesting,'' in
  \emph{Proc., IEEE Asilomar}, Monterey, CA, Nov. 2011.

\bibitem{DhiGanC2011}
H.~S. Dhillon, R.~K. Ganti, and J.~G. Andrews, ``A tractable framework for
  coverage and outage in heterogeneous cellular networks,'' in \emph{Proc.,
  Information Theory and its Applications (ITA)}, San Diego, CA, Feb. 2011.

\bibitem{DhiGanJ2012}
H.~S. Dhillon, R.~K. Ganti, F.~Baccelli, and J.~G. Andrews, ``Modeling and
  analysis of {K}-tier downlink heterogeneous cellular networks,'' \emph{IEEE
  Journal on Sel. Areas in Communications}, vol.~30, no.~3, pp. 550 -- 560,
  Apr. 2012.

\bibitem{JoSanJ2012}
H.-S. Jo, Y.~J. Sang, P.~Xia, and J.~G. Andrews, ``Heterogeneous cellular
  networks with flexible cell association: A comprehensive downlink {SINR}
  analysis,'' \emph{IEEE Trans. on Wireless Communications}, vol.~11, no.~10,
  pp. 3484 -- 3495, Oct. 2012.

\bibitem{MukJ2012}
S.~Mukherjee, ``Distribution of downlink {SINR} in heterogeneous cellular
  networks,'' \emph{IEEE Journal on Sel. Areas in Communications}, vol.~30,
  no.~3, pp. 575 -- 585, Apr. 2012.

\bibitem{MadResC2011}
P.~Madhusudhanan, J.~G. Restrepo, Y.~Liu, T.~X. Brown, and K.~R. Baker,
  ``Multi-tier network performance analysis using a shotgun cellular system,''
  in \emph{Proc., IEEE Globecom}, Houston, TX, Dec. 2011.

\bibitem{TayDhiC2012}
D.~B. Taylor, H.~S. Dhillon, T.~D. Novlan, and J.~G. Andrews, ``Pairwise
  interaction processes for modeling cellular network topology,'' in
  \emph{Proc., IEEE Globecom}, Anaheim, CA, Dec. 2012.

\bibitem{BlaKarJ2012}
B.~Blaszczyszyn, M.~K. Karray, and H.-P. Keeler, ``Using {Poisson} processes to
  model lattice cellular networks,'' available online: arxiv.org/abs/1207.7208.

\bibitem{ElSHosJ2013}
H.~ElSawy, E.~Hossain, and M.~Haenggi, ``Stochastic geometry for modeling,
  analysis, and design of multi-tier and cognitive cellular wireless networks:
  A survey,'' {\em IEEE Communications Surveys and Tutorials}, to appear, 2013.

\bibitem{DhiAndJ2013}
H.~S. Dhillon and J.~G. Andrews, ``Downlink rate distribution in heterogeneous
  cellular networks under generalized cell selection,'' submitted to {\em IEEE
  Wireless Communications Letters}, Jun. 2013. Available online:
  arxiv.org/abs/1306.6122.

\bibitem{SinDhiJ2013}
S.~Singh, H.~S. Dhillon, and J.~G. Andrews, ``Offloading in heterogeneous
  networks: Modeling, analysis and design insights,'' \emph{IEEE Trans. on
  Wireless Communications}, vol.~12, no.~5, pp. 2484 -- 2497, May 2013.

\bibitem{RouWriB2004}
S.~Roundy, P.~K. Wright, and J.~M. Rabaey, \emph{Energy Scavenging for Wireless
  Sensor Networks: With Special Focus on Vibrations}.\hskip 1em plus 0.5em
  minus 0.4em\relax Norwell, MA: Kluwer Academic Publishers, 2004.

\bibitem{KinB1993}
J.~F.~C. Kingman, \emph{Poisson Processes}.\hskip 1em plus 0.5em minus
  0.4em\relax Oxford University Press, 1993.

\bibitem{DhiGanJ2013b}
H.~S. Dhillon, R.~K. Ganti, and J.~G. Andrews, ``Modeling non-uniform {UE}
  distributions in downlink cellular networks,'' \emph{IEEE Wireless
  Communications Letters}, vol.~2, no.~3, pp. 339 -- 342, Jun. 2013.

\bibitem{BacBlaB2009}
F.~Baccelli and B.~Blaszczyszyn, \emph{Stochastic Geometry and Wireless
  Networks, Volume I -- Theory}.\hskip 1em plus 0.5em minus 0.4em\relax NOW:
  Foundations and Trends in Networking, 2009.

\bibitem{KeeBlaC2013}
H.-P. Keeler, B.~Blaszczyszyn, and M.~K. Karray, ``{SINR-based} $k$-coverage
  probability in cellular networks with arbitrary shadowing,'' in \emph{Proc.,
  IEEE Intl. Symposium on Information Theory}, Istanbul, Jul. 2013.

\bibitem{TarJ1955}
A.~Tarski, ``A lattice-theoretical fixpoint theorem and its applications,''
  \emph{Pacific J. Math.}, vol.~5, no.~2, pp. 285 -- 309, 1955.

\bibitem{KenJ2001}
J.~Kennan, ``Uniqueness of positive fixed points for increasing concave
  functions on $\mathbb{R}^n$: An elementary result,'' \emph{Review of Economic
  Dynamics}, vol.~4, no.~4, pp. 893 -- 899, Oct. 2001.

\bibitem{ResB2005}
S.~I. Resnick, \emph{Adventures in Stochastic Processes}.\hskip 1em plus 0.5em
  minus 0.4em\relax Boston: Birkh\"{a}user, 2005.

\bibitem{3GPPM2012a}
{\em Evolved Universal Terrestrial Radio Access {(E-UTRA)} and Evolved
  Universal Terrestrial Radio Access Network {(E-UTRAN)}; Overall description;
  Stage 2}, 3GPP TS 36.300, Jul. 2012.

\bibitem{StoKenB1995}
D.~Stoyan, W.~S. Kendall, and J.~Mecke, \emph{Stochastic Geometry and Its
  Applications}, 2nd~ed.\hskip 1em plus 0.5em minus 0.4em\relax Chichester:
  John Wiley and Sons, 1995.

\bibitem{DhiKouJ2013}
H.~S. Dhillon, M.~Kountouris, and J.~G. Andrews, ``Downlink {MIMO} {HetNets}:
  Modeling, ordering results and performance analysis,'' submitted to {\em IEEE
  Tran. Wireless Communications}, Jan. 2013. Available online:
  arxiv.org/abs/1301.5034.

\bibitem{LiZhaJ2013}
C.~Li, J.~Zhang, and K.~B. Letaief, ``Throughput and energy efficiency analysis
  of small cell networks with multi-antenna base stations,'' submitted to {\em
  IEEE Tran. Wireless Communications}, 2013. Available online:
  arxiv.org/abs/1306.6169.

\bibitem{NovDhiJ2012}
T.~D. Novlan, H.~S. Dhillon, and J.~G. Andrews, ``Analytical modeling of uplink
  cellular networks,'' \emph{IEEE Trans. on Wireless Communications}, vol.~12,
  no.~6, pp. 2669 -- 2679, Jun. 2013.

\bibitem{BacBreB2003}
F.~Baccelli and P.~Br\'{e}maud, \emph{Elements of queueing theory: palm
  martingale calculus and stochastic recurrences}.\hskip 1em plus 0.5em minus
  0.4em\relax Springer-Verlag, 2003.

\end{thebibliography}

\end{document}